\documentclass[journal, twocolumn]{IEEEtran}
\IEEEoverridecommandlockouts
\usepackage{amsmath,amssymb,amsfonts}
\usepackage{graphicx}
\usepackage{textcomp}
\usepackage{xcolor}
\usepackage{booktabs}
\usepackage{algorithm}
\usepackage{siunitx}
\usepackage{amsthm}
\usepackage{multirow}
\usepackage{algpseudocode}

\usepackage[sort&compress, numbers]{natbib}
\usepackage[hidelinks]{hyperref}
\usepackage{caption}
\newtheorem{theorem}{Theorem}
\def\BibTeX{{\rm B\kern-.05em{\sc i\kern-.025em b}\kern-.08em
    T\kern-.1667em\lower.7ex\hbox{E}\kern-.125emX}}

\begin{document}

\title{Channel-Aware Distributed Transmission Control and Video Streaming in UAV Networks}

\author{

\IEEEauthorblockN{Masoud Ghazikor\IEEEauthorrefmark{1}, Keenan Roach\IEEEauthorrefmark{2}, Kenny Cheung\IEEEauthorrefmark{2}, Morteza Hashemi\IEEEauthorrefmark{1}}
        \\
        \IEEEauthorblockA{
        \IEEEauthorrefmark{1}Department of Electrical Engineering and Computer Science, University of Kansas, \\
        \IEEEauthorrefmark{2}Universities Space Research Association (USRA), \\
        }
}

\maketitle
\thispagestyle{plain}
\pagestyle{plain}
\begin{abstract}
In this paper, we study the problem of distributed transmission control and video streaming optimization for unmanned aerial vehicles (UAVs) operating in unlicensed spectrum bands. We develop a rigorous cross-layer analysis framework that \emph{jointly} considers three inter-dependent factors: (i) in-band interference introduced by ground-aerial nodes at the physical (PHY) layer, (ii) limited-size queues with delay-constrained packet arrival at the medium access control (MAC) layer, and (iii) video encoding rate at the application layer. First, we formulate an optimization problem
to maximize the average throughput by optimizing the fading threshold (transmission policy). To this end, we jointly analyze the queue-related packet loss probabilities  (i.e., buffer overflow and time threshold event) as well as the outage probability due to the low signal-to-interference-plus-noise ratio (SINR). We introduce the Distributed Transmission Control (DTC) algorithm that maximizes the average throughput by adjusting transmission policies to balance the trade-offs between packet drop from queues vs. transmission errors due to low SINRs.
Second, we incorporate the video distortion model to develop distributed peak signal-to-noise ratio (PSNR) optimization for video streaming. The formulated optimization incorporates two cross-layer parameters, specifically the fading threshold and video encoding rate. To tackle this problem, we develop the Joint Distributed Video Transmission and Encoder Control (JDVT-EC) algorithm that enhances the average PSNR for all nodes by fine-tuning transmission policies and video encoding rates to balance the trade-offs between packet loss and lossy video compression distortions. Through extensive numerical analysis, we thoroughly examine the proposed algorithms and demonstrate that they are able to find the optimal transmission policies and video encoding rates under various scenarios. Notably, our approach improves the average throughput by 1.7\% to 51.65\% compared to various baselines, including the selfish and random policies. Additionally, we demonstrate an average PSNR increase of 0.24 dB and 1.7 dB compared to separately optimizing the fading threshold and video encoding rate, respectively.
\end{abstract}

\begin{IEEEkeywords}
UAVs, Unlicensed spectrum, Distributed transmission control, Distributed video streaming optimization.
\end{IEEEkeywords}

\section{Introduction}
Over the past few years, unmanned aerial vehicle (UAV)-based wireless communication has shown promising applications due to ease of deployment, high agility, ability to perform different tasks, and the high probability of establishing Line-of-Sight (LoS) links with ground-aerial nodes. In general, wireless networks with mobile and dynamic topologies, such as UAV networks, require efficient resource allocation and network optimization strategies to handle interference, queue dynamics, and network congestion. Traditional centralized techniques (e.g., ~\cite{yin2015joint, Nishiyama-2017-OFDM}) work well in \emph{infrastructure-based} settings, but they are typically ineffective when no central infrastructure is available for coordination across distributed nodes. This scenario is common when multiple ground and aerial nodes operate in an \emph{infrastructure-less} mode using unlicensed spectrum bands, as shown in Fig. \ref{system_model}. In this case, unlicensed bands are shared among distributed nodes that are more susceptible to interference, thereby creating an interference-limited network with ground and aerial nodes~\cite{Chintareddy-2023-Collaborative}.

\begin{figure}[t]
\includegraphics[width=0.9\linewidth]{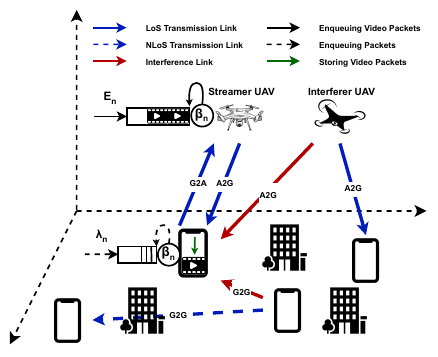}
\centering
\caption{System model including aerial and ground nodes operating in unlicensed spectrum bands. The proposed system model captures in-band interference, limited-size queues with delay-constrained traffic, and video streaming between the streamer UAV and its associated ground node.
}
\label{system_model}
\end{figure}

In the context of infrastructure-less networks, distributed optimization solutions have been proposed in prior works (see, for example, ~\cite{scutari2008optimal,yang2016distributed,guan2011distributed} and references therein) to enable wireless nodes to make local decisions (transmit power, precoding matrices, sub-channel selection, etc.) based on network state information. However, UAV networks introduce unique challenges, such as dynamic link conditions, distinct interference patterns compared with ground networks, and stringent performance requirements, that necessitate developing distributed optimization solutions that are specifically tailored for their operation scenarios. For instance, a common deployment scenario for an infrastructure-less, interference-limited UAV network is video streaming services in areas with limited physical infrastructure, such as disaster-affected areas. In fact, the capability of UAVs for video streaming plays a crucial role in public safety applications~\cite{Khan-2024-Efficient}, as they can be deployed in search and rescue operations to efficiently cover large areas and provide live videos to enable quick decision-making and resource allocation strategies~\cite{Yuan-2017-ComProSe}.

Hence, developing \emph{distributed transmission control and video streaming policies} is crucial to ensure high quality of experience (QoE), particularly for UAVs that depend on time-sensitive video packets and command-and-control (C2) data~\cite{Mohit-2023-UAV, Badnava-2022-QoE}. Although extensive research has been conducted on UAV networks~\cite{Bithas-2023-Generalized,Bithas-2020-Channel,Cui-2019-Data} (more details in Section \ref{Related Work}), development of a comprehensive system model and problem formulation for distributed transmission policies that account for, interference effects, queue dynamics, and video streaming requirements is still lacking.

To address this gap, we develop distributed transmission policies for integrated ground-aerial networks by jointly considering: (i) interference levels in unlicensed spectrum bands and their impact on the signal-to-interference-plus-noise ratio (SINR), (ii) transmission queue states in terms of buffer size and queuing delay, and (iii) video encoding rate optimization. By considering the interplay of these interrelated components, we explore how UAVs and ground-level nodes achieve an optimal policy by adjusting the fading threshold and video encoding rate in a distributed manner to achieve optimal average throughput and peak signal-to-noise ratio (PSNR). 

In this paper, we formulate two optimization problems to maximize the average throughput and PSNR performance, which is an objective metric for measuring distortion in video streaming applications, and has been widely used in the literature due to various reasons such as being easy to calculate, having a clear physical meaning, and being mathematically convenient for optimization~\cite{Zhou-2004-Image, Xiaoqing-2005-Congestion, Tian-2016-Interference, BadnavaCH:24a}. To solve the formulated problems, we first propose the Distributed Transmission Control (DTC) algorithm that determines the optimal fading threshold set $\boldsymbol{\beta}^{\star}$ in the environment to maximize the average throughput for all nodes. Next, to solve the video streaming problem, we develop the Joint Distributed Video Transmission and Encoder Control (JDVT-EC) algorithm that optimizes the average PSNR for all nodes by adjusting fading thresholds and video encoding rates to balance the trade-offs between packet loss and lossy video compression distortions. Table \ref{proposed_algorithms} provides a summary of the proposed algorithms, including their objective functions and optimization parameters. Through comprehensive numerical results, we compare the performance of our proposed algorithms with several baselines, confirming the effectiveness of our method. In summary, the main contributions of this paper are as follows. 

\begin{itemize}
\item \emph{Cross-Layer optimization framework:}  We develop a comprehensive cross-layer optimization framework that explicitly integrates transmission errors due to low SINR, queue-related packet loss probabilities (i.e., buffer overflow and time threshold losses), and video compression distortions into a unified problem formulation. 
Our framework is designed to optimize the average throughput and PSNR by adjusting fading thresholds and video encoding rates for an integrated aerial-ground network operating in the unlicensed spectrum bands.

\item \emph{Distributed algorithms for transmission control and video streaming:} Using consensus-based distributed algorithm and coordinate descent optimization, we develop two algorithms: (i) Distributed Transmission Control (DTC) algorithm, which dynamically adjusts fading thresholds to maximize the average throughput by mitigating trade-offs between low-SINR transmission errors and queue packet losses, and (ii) Joint Distributed Video Transmission and Encoder Control (JDVT-EC) algorithm, which optimally balances packet loss probabilities and video distortions by jointly adjusting fading thresholds and video encoding rates. These algorithms provide the optimal fading threshold and video encoding rate for maximizing the average throughput and PSNR across all aerial and ground nodes.

\item \emph{Extensive performance evaluation:} 
We conduct extensive numerical evaluations to demonstrate the efficacy of our proposed algorithms compared with several baselines. In particular, we achieve up to $51.65\%$ higher average throughput and up to $1.7$ dB improvement in average PSNR compared to several baselines, such as selfish policy, random policy, or separately optimized parameters. Our numerical results demonstrate that the proposed cross-layer and distributed algorithms achieve enhanced trade-offs between different factors that contribute to the packet losses and video distortion.
\end{itemize}
\noindent The rest of this paper is organized as follows. In Section \ref{Related Work}, the related work is presented. Section \ref{System Model} describes the system model, including wireless channels, queuing, and video streaming models. In Section \ref{Bi-Directional Link Performance}, we formulate the distributed transmission control problem in terms of queue and interference analysis for throughput optimization. Section \ref{Proposed Solution} investigates the distortion-based PSNR optimization problem. In Section \ref{Numerical Results}, extensive numerical results are provided, followed by the conclusion in Section \ref{Conclusion}. 

\begin{table}
\centering
\caption{Overview of Proposed Optimization Algorithms}
\resizebox{\columnwidth}{!}{
\begin{tabular}{l c c}
    \toprule
    Algorithm & Decision Variable & Objective Function \\
    \midrule
    DTC (Alg. \ref{alg:DTC}) & Fading Threshold & Average Throughput \\
    JDVT-EC (Alg. \ref{alg:JDVT-EC}) & Fading Threshold $\&$ Encoding Rate & Average PSNR \\
    \midrule
    \multicolumn{3}{l}{Sub-algorithms of JDVT-EC} \\
    \quad DVTC (Alg. \ref{alg:DVTC}) & Fading Threshold & Average PSNR \\
    \quad DVEC (Alg. \ref{alg:DVEC}) & Encoding Rate & Average PSNR \\
    \bottomrule
\end{tabular}
}
\label{proposed_algorithms}
\end{table}

\section{Related Work} \label{Related Work}
We discuss related works by dividing our analysis into three main categories: distributed transmission control, interference-limited UAV networks, and video streaming $\&$ PSNR-distortion optimization.

\textbf{Distributed Transmission Control.} In this part, we review those related works that are focused on developing distributed transmission control policies. \citeauthor{Guan-2016-ToTransmit}~\cite{Guan-2016-ToTransmit} proposed a distributed transmission policy for adjusting fading thresholds in interference-limited networks. They developed a game-theoretic solution to maximize throughput by considering transmission errors in the ground-to-ground (G2G) channel and queuing delays in the transmission queue. Furthermore, the authors in~\cite{Tian-2016-Interference} proposed a cross-layer framework by considering transmission error, queuing delay, and encoding rate in a cross-layer manner for the G2G channel. They formulated a distributed optimization problem to maximize PSNR and solved it using game theory. The authors in~\cite{Yi-2018-Transmission} proposed a framework for managing delay-sensitive packet transmissions to base stations (BSs) by considering packet priorities. A non-cooperative delay-dependent prioritized queuing game was formulated to optimize distributed transmission scheduling. In addition to these works, several other works have studied distributed transmission control~\cite{Guan-2016-ToTransmit, Yi-2018-Transmission, Tian-2018-QoS}, especially for video packet transmission \cite{Tian-2016-Interference, Ye-2018-Quality}. However, these studies do not consider aerial networks, which present completely different physical channel and interference models. In this paper, we develop distributed cross-layer algorithms that explicitly capture the characteristics of both ground and aerial links.

\textbf{Interference-limited UAV Networks.} In this category of related works, \citeauthor{Hellaoui-2019-Towards}~\cite{Hellaoui-2019-Towards} proposed a cellular-based UAV control framework in which a BS transmits C2 messages toward a UAV in the presence of interference nodes. They introduced an optimal sub-carrier allocation algorithm to compensate for the low throughput due to interference nodes. The authors in~\cite{Shen-2020-Multi} considered a system model with interference channels such that multiple UAVs communicate with ground nodes, while managing interference due to LoS links. Furthermore, ~\cite{Kim-2019-Impact} analyzed the impact of interference on UAVs by considering different channel fading for LoS and NLoS links. A closed-form outage probability expression is derived, revealing the existence of an optimal UAV height that minimizes outage. On top of these studies, extensive amounts of work have focused on different aspects of interference-limited UAV communications~\cite{Hellaoui-2019-Towards, Shen-2020-Multi, Kim-2019-Impact, Mei-2020-Cooperative, Kim-2018-Outage}. While these studies have explored interference-limited UAV communications, they do not account for \emph{distributed decision-making} across ground-aerial nodes. This paper aims to address this gap.

\textbf{Video Streaming $\&$ PSNR-Distortion Optimization.} In the context of video streaming and PSNR-distortion optimization, the authors in~\cite{He-2019-QoE} addressed the uplink bandwidth allocation problem for UAVs performing real-time video streaming. A distributed allocation mechanism is proposed to optimize video quality, allowing each UAV to independently adjust its video encoding rate. 
Using a self-learning algorithm, UAVs iteratively update their bandwidth strategies until reaching an equilibrium. 
More recently,~\cite{Kesong-2024-Joint} proposed the joint video coding and channel transmission optimization for UAVs. This study developed an integrated delay-power-rate-distortion model to capture the correlation between video coding and channel transmission and minimize end-to-end distortions and UAV power consumption. Furthermore, the authors in~\cite{Shen-2024-UAV} proposed a UAV-assisted video multicast streaming for non-orthogonal multiple access (NOMA) networks to improve aggregate PSNR for cell-edge users. A graph-based model is developed to optimize UAV placement, multicast group association, and subchannel allocation. In addition to these studies, there are additional studies on UAV-based video streaming \cite{He-2019-QoE, Yu-2021-Multi, Zhan-2021-Joint, Liao-2022-QoE} and PSNR-distortion optimization for general networks (i.e., not necessarily UAV networks)~\cite{Tian-2016-Interference, Kesong-2024-Joint, Shen-2024-UAV, Tang-2021-QoE, BadnavaCH:24a}. Compared with our paper, these prior works did not specifically focus on distributed optimization for \emph{interference-limited} ground-aerial networks.

Overall, while prior studies have provided significant insights on distributed algorithms, interference-limited UAV networking, and video streaming optimization, our work advances these efforts by introducing a unified framework that integrates and bridges the gap between these three components. In particular, this paper develops distributed transmission control and video streaming policies that jointly consider \textbf{(i)} in-band interference in unlicensed spectrum bands at the PHY layer, \textbf{(ii)} the transmission queue condition in terms of buffer overflow and maximum queuing delay at the MAC layer, and \textbf{(iii)} the level of the video encoding rate at the application layer. This paper extends our preliminary results in \cite{ghazikor2024interference} in several directions. In the preliminary work, we only investigated the distributed transmission control policy by adjusting the fading threshold. In this paper, we extend the distributed transmission control policy to the distributed video transmission and encoder control policy, which jointly adjusts the fading threshold and video encoding rate to maximize the average PSNR over distributed ground-aerial nodes using a cross-layer scheme at application, MAC, and PHY layers. Furthermore, we have significantly extended our numerical evaluations to investigate throughput and PSNR optimization, as well as the spatial performance of the proposed algorithms.

\section{System Model} \label{System Model}
Our envisioned system model is shown in Fig. \ref{system_model}, which consists of ground and aerial nodes communicating over unlicensed spectrum bands to establish ground-to-air (G2A), air-to-ground (A2G), ground-to-ground (G2G), and air-to-air (A2A) links. The spectrum band is divided into a set of $\boldsymbol{F}$ frequency channels, such that
$\boldsymbol{N}$ denotes the set of communication sessions that share the same spectrum band, and $n \in \boldsymbol{N}$ represents the individual session between the source node and streamer UAV. We assume that the source node communicates with the {streamer UAV} (left UAV in Fig. \ref{system_model}), while one of the interferer ground nodes communicates with the interferer UAV (right UAV in Fig. \ref{system_model}). Each node dynamically adjusts its transmission policy and encoding rate to optimize throughput and video quality. To this end, we incorporate the impacts of three key components, namely (i) wireless channel models, which capture the characteristics of different links, (ii) queue model, which accounts for finite buffer sizes and delay-sensitive data traffic, and (iii) video streaming model, which controls encoding rates to balance video distortion and network congestion. Unlike existing UAV communication models (as presented in Section \ref{Related Work}), our system model explicitly integrates these components to dynamically find optimal fading thresholds and video encoding rates based on channel conditions, queue states, and interference levels. This cross-layer approach enables aerial and ground nodes to adjust their transmission policies and encoding rates in a distributed manner.

\begin{table}
    \centering
    \caption{Key Symbols}
    \resizebox{\columnwidth}{!}{%
    \label{tab:key_symbols}
    \begin{tabular}{lc}
        \toprule
        Definition & Symbol \\
        \midrule
        LoS probability environment parameters & $\zeta$, $v$, $\mu$\\
        Source node index & $n$\\
        Interferer node index & $m$\\
        Node index & $i$\\
        Elevation angle & $\theta_i$\\
        Square root of single-slope path loss & $\hat{h}_i^f$\\
        Channel fading coefficient & $\Tilde{h}_i^f$\\
        Probability Density  Function & $\mathrm{f}_{X}(x)$ \\
        Cumulative Density Function & $\mathcal{F}_{X}(x)$ \\ 
        LoS and NLoS path loss exponents & $\alpha_{L}$, $\alpha_{N}$\\
        Reference and node-destination distances & $d_0$, $d_i$\\
        Number of sub-channels & $|\boldsymbol{F}|$\\
        LoS and NLoS Rician factors & $K_{L}$, $K_{N}$\\
        Channel fading threshold & $\beta_i$\\
        Time threshold and time slot duration & $T^{th}_n$, $T^{slt}_n$\\
        Average incoming packet rate & $\lambda_n$\\
        Normalized buffer capacity & $\Tilde{b}_n$\\
        Transmission power & $P_i$\\
        SINR threshold & $\gamma_{th}$\\
        Thermal noise power & $\sigma^2$\\
        Operating frequency and bandwidth & $f$, $W$\\
        Noise temperature & $T$\\
        Boltzmann constant & $k$\\
        Sensitivity parameter & $s_n$\\
        Video encoding rate & $E_n$\\
        Average video packet length & $L_n$\\
        Rate-distortion parameters & $D_{0}$, $E_{0}$, $\theta_0$\\
        Video pixel bit-depth & $p$\\ 
        \bottomrule
    \end{tabular}
    }
\end{table}

\subsection{Channel Model}
\noindent 
\textbf{Distance-Based LoS Probability Model.}
In our previous studies~\cite{Ghazikor-2023-Exploring}, we focused on G2A and A2G channels, deriving LoS probability $\mathbb{P}_L(\theta_i)$ and Rician factor $K(\theta_i)$ according to the elevation angle $\theta_i$ \cite{Kim-2018-Outage}. However, this approach is not suitable for A2A channels since, as the altitude of the nodes increases, $\theta_i$ decreases, leading to reduced values of $\mathbb{P}_L(\theta_i)$ and $K(\theta_i)$ for A2A channels. In reality, we expect an increase in both $\mathbb{P}_L(\theta_i)$ and $K(\theta_i)$ at higher altitudes. Therefore, we present a complete LoS probability model based on distance $d_i$, covering all G2G, G2A, A2G, and A2A channels. Thus, $\mathbb{P}_L(d_i)$ can be expressed as follows \cite{Mohammed-2021-Line, Kim-2019-Impact}:
\begin{align*}
\mathbb{P}_{L}(d_i) = 
\begin{cases}
    \Bigl(1-e^{(-\frac{z_{i}^2}{2\zeta^2})}\Bigl)^{d_i\sqrt{v\mu}}       & \text{$z_{i} = z_{u}$}\\
    \Bigl(1-\frac{\sqrt{2\pi}\zeta}{d_i^{V}}\left|Q(\frac{z_{i}}{\zeta})-Q(\frac{z_{u}}{\zeta})\right|\Bigl)^{d_i^{H}\sqrt{v\mu}}          & \text{$z_{i} \neq z_{u}$}.
\end{cases}
\end{align*}
Here, $\zeta$, $v$, and $\mu$ represent environmental parameters, and $Q(.)$ denotes the $Q$-function. Furthermore, $d_i^{H} = \sqrt{(x_{i}-x_{u})^2 + (y_{i}-y_{u})^2}$ and $d_i^{V} = \sqrt{(z_{i}-z_{u})^2}$ denote the horizontal and vertical distances between the transmitter node $i$ and the receiver node $u$, respectively. Hence, the total distance between node $i$ and a designated receiver is determined by $d_i = \sqrt{{d_i^H}^2 + {d_i^V}^2}$ for $\forall{i} = \left\{n,\boldsymbol{m}\right\}$, where the indices $n \in \boldsymbol{N}$ and $\boldsymbol{m} \subset \boldsymbol{N}$ are the source node and the set of interferer nodes within a specified area, respectively.

\noindent
\textbf{Single-Slope Path Loss Model.}
With the transmit power $P_t$, the received power $P_r$ is expressed as $P_r = P_t |h_n^f|^2$, where $h_n^f$ denotes the channel gain of sub-channel $f \in \boldsymbol{F}$. Also, $h_n^f$ can be defined as $h_n^f = \Tilde{h}_n^f \hat{h}_n^f$, where $\Tilde{h}_n^f$ and $\hat{h}_n^f$ are the fading coefficient and the square root of the path loss, respectively. By the single-slope path loss model~\cite{Ren-2011-Modelling}, we have:
\begin{align}
\hat{h}_n^f = \sqrt{C(\frac{d_0}{d_n})^{\alpha(d_n)}} \quad \text{if} \enskip d_n \ge d_0, 
\end{align}
where $C = \frac{\lambda^2}{16\pi^2d_0^2}$ serves as a constant factor in which $\lambda = \frac{c}{f}$ represents the wavelength. Furthermore, $d_0$ and $d_n$ express the reference distance and the distance between the source and destination nodes, respectively. Furthermore, the path loss exponent is defined as $\alpha(d_n) = \alpha_{L}\mathbb{P}_{L}(d_n) + \alpha_{N}(1 - \mathbb{P}_{L}(d_n))$, where $\alpha_{L}$ and $\alpha_{N}$ denote the path loss exponents for LoS and NLoS links, respectively \cite{Azari-2018-Ultra}.

\subsection{Queuing Model}
We assume each node has a limited-size queue, where the packet arrival process ($\lambda_i$) follows a Poisson distribution. Each node either transmits the packet to its destination or keeps it in its queue. This transmission decision is determined based on channel conditions and queue states. For instance, if two or more nodes choose the same channel to transmit packets simultaneously, there would be in-band interference and degraded SINR values. Packets may be dropped from the queue due to two primary events: (i) buffer overflow and (ii) time threshold loss that occurs if a packet remains in the queue beyond a time threshold $T_n^{\text{th}}$. The probabilities of packet loss due to buffer overflow and time threshold are, respectively, denoted by $P_n^{\text{ov}}(\beta_n)$ and $P_n^{\text{dly}}(\beta_n)$.  These queue-related losses directly impact video quality, as they increase overall distortion, as detailed in Section~\ref{Proposed Solution}.

To develop distributed transmission policies, it is important to capture the interplay between queue dynamics and transmission strategies. In particular, a higher fading threshold $\beta_n$ reduces transmission errors but increases the probability of packet loss due to buffer overflow and exceeding time threshold $T_n^{\text{th}}$. Our proposed DTC optimization framework balances this trade-off by dynamically adjusting $\beta_n$ based on queue states and interference conditions in order to improve the average throughput. Furthermore, our proposed JDVT-EC algorithm jointly optimizes with respect to the fading threshold $\beta_n$ and video encoding rate $E_n$, as described next.

\subsection{Video Streaming Model}
In our UAV network, we aim to design an efficient video streaming policy that follows a cross-layer framework, where decisions made at the application, MAC, and PHY layers impact the overall streaming quality. 
In particular, each streaming node encodes video packets at the bitrate $E_n$, which impacts video distortions and network congestion. A higher encoding rate enhances the PSNR by reducing video distortions, but at the cost of increased network traffic and interference for other nodes, which in turn, can increase packet losses from the queues as controlled by adjusting the fading threshold $\beta_n$. To ensure high-quality video streaming, we consider an adaptive encoding strategy where the UAV streaming node adjusts the encoding rate $E_n$ based on link conditions and available transmission resources.

Once video frames are encoded, they are packetized and added to a finite-size transmission queue. As mentioned, each packet is subject to three types of losses: \textbf{(i)} buffer overflow loss, \textbf{(ii)} time threshold loss, and \textbf{(iii)} transmission error loss if the SINR is less than the threshold $\gamma_{th}$. To mitigate these losses, our proposed \emph{Joint Distributed Video Transmission and Encoder Control (JDVT-EC) algorithm} dynamically adjusts both the video encoding rate $E_n$ and transmission policy characterized in terms of fading threshold $\beta_n$.
The detailed mathematical formulation of the video distortion model and PSNR optimization are presented in Section \ref{Proposed Solution}.

With the system model established, we now define the optimization framework that enables distributed nodes to optimally adjust their fading thresholds and video encoding rates. First, we formalize the distributed transmission control for adjusting fading thresholds in Section \ref{Bi-Directional Link Performance}, followed by the joint optimization of fading thresholds and encoding rates in Section \ref{Proposed Solution}.

\section{Distributed Transmission Control} \label{Bi-Directional Link Performance}
In this section, we characterize the throughput performance in terms of constituent queuing and interference components. Next, we present the Distributed Transmission Control (DTC) algorithm to optimize the average throughput in a distributed manner.

\begin{figure*}[t]
    \centering
    \begin{minipage}{.33\linewidth}
        \includegraphics[width=6cm, trim={30 0 80 40}, clip]{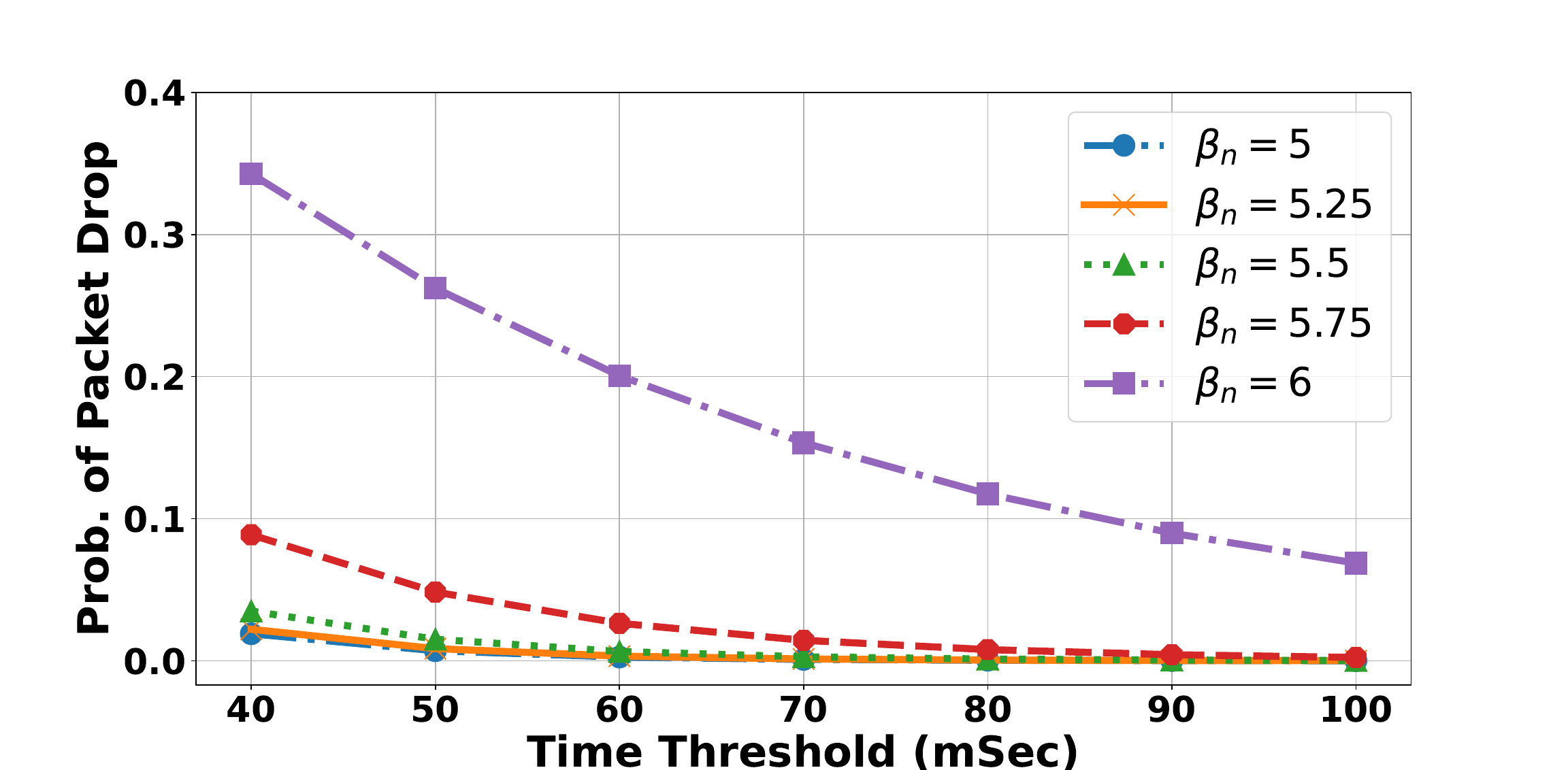}
        \caption{$\mathbb{P}_n^{dly}(\beta_n)$ vs. $T^{th}_n$ by Different $\beta_n$.}
        \label{Pdly_Tth}
    \end{minipage}%
    \begin{minipage}{.33\linewidth}
        \includegraphics[width=6cm, trim={30 0 80 40}, clip]{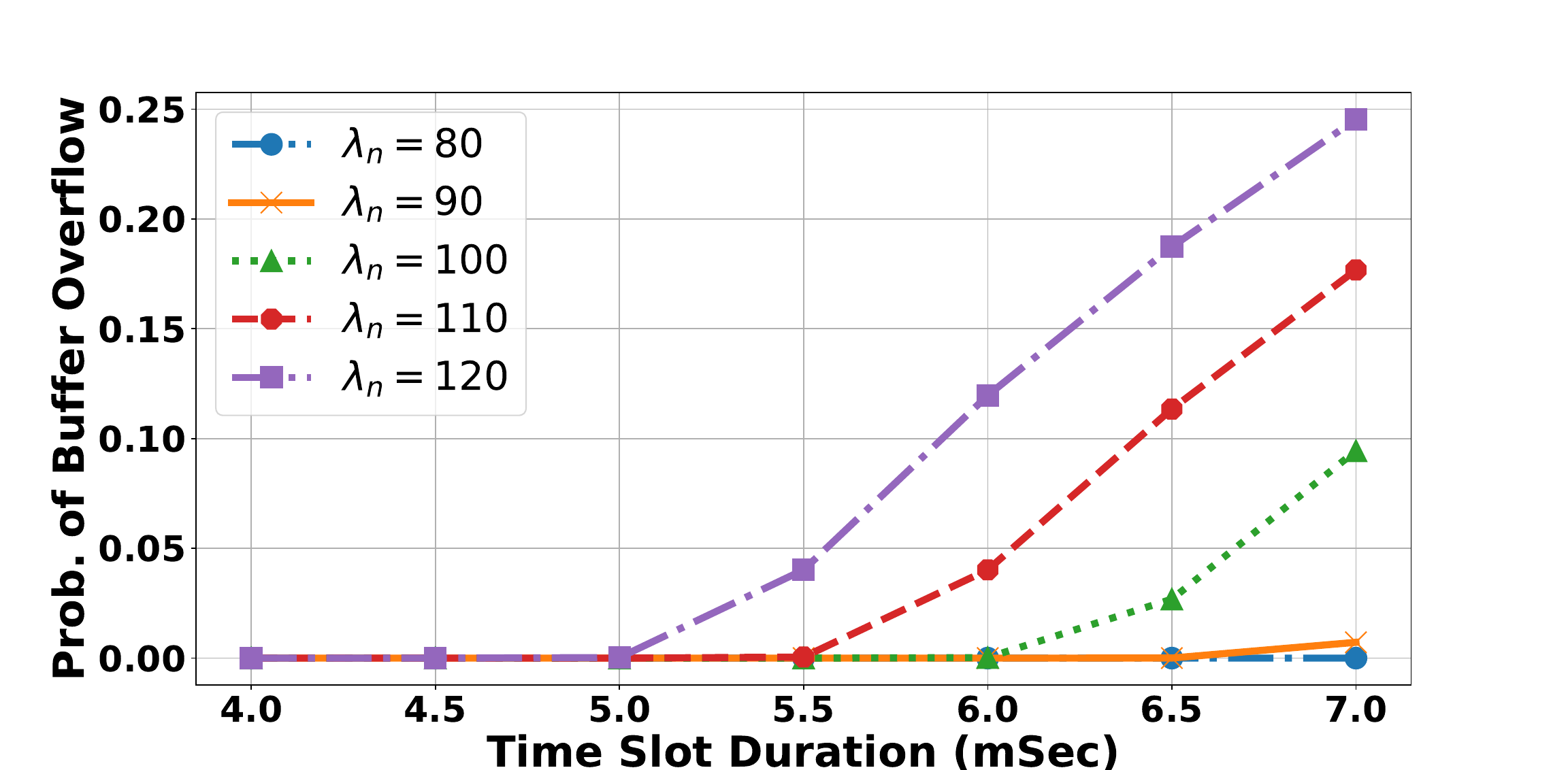}
        \caption{$\mathbb{P}_n^{ov}(\beta_n)$ vs. $T^{slt}_n$ by Different $\lambda_n$.}
        \label{Pov_Tslt}
    \end{minipage}%
    \begin{minipage}{.33\linewidth}
        \includegraphics[width=6cm, trim={30 0 80 40}, clip]{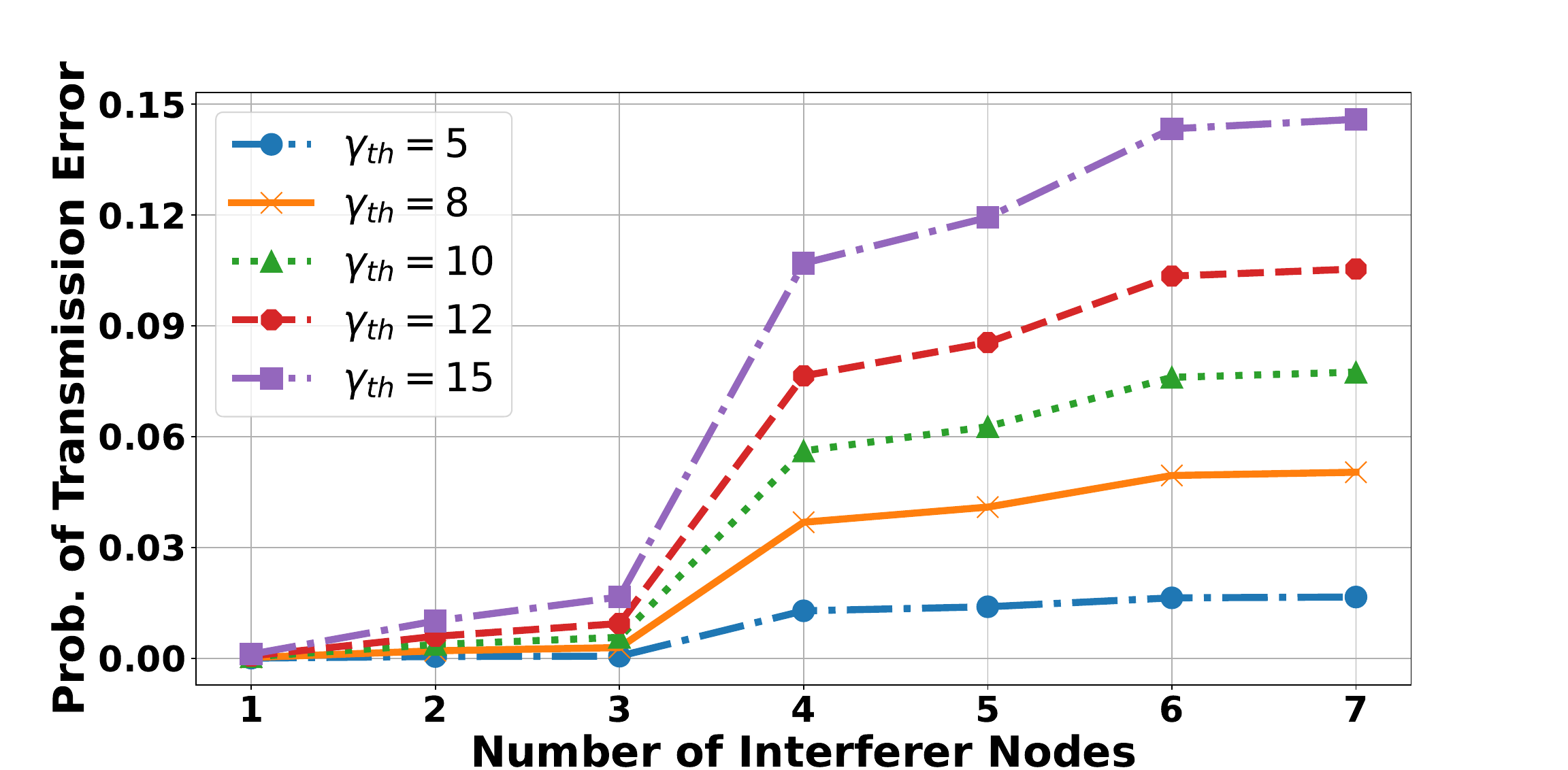}
        \caption{$\mathbb{P}_n^{err}(\boldsymbol{\beta})$ vs. \# of Interferers by Different $\gamma_{th}$.}
        \label{Perr_Nodes}
    \end{minipage}%
\end{figure*}

\subsection{Channel Transmission Policy}
In the context of a block fading channel model, the variable $\Tilde{h}_{n}^f$ follows Rician or Rayleigh distributions, depending on whether it corresponds to LoS or NLoS channels, respectively \cite{Kumar-2023-Outage}. Here, we only discuss the Rician fading model for brevity~\cite{Ghazikor-2023-Exploring}. Hence, the probability density function (PDF) of the Rician distribution for the LoS channel is given by \cite{Zhou-2018-Coverage}:
\begin{align}
\mathrm{f}_{\Tilde{h}_{n}^f}(x) = xe^{-\frac{x^2+b^2}{2}}I_0(xb), 
\end{align}
where $I_0$ denotes the modified Bessel function of the first kind with order zero. Also, the parameter $b = \sqrt{2K(d_n)}$ is defined by the Rician factor $K(d_n) = K_{N}e^{\ln(\frac{K_{L}}{{K_{N}}})\mathbb{P}_{L}(d_n)^2}$, in which $K_{L}$ and $K_{N}$ are determined based on the condition that $\mathbb{P}_{L}(d_n)$ is equal to one and zero, respectively \cite{Kim-2019-Impact}. 

The source node transmits its packet to the destination node over the best frequency channel $f^* = {\arg\max}_{f \in \boldsymbol{F}} \ \ \Tilde{h}_n^f \hat{h}_n^f$ if the fading coefficient is greater than the fading threshold $\beta_n$, i.e., $\Tilde{h}_n^{f^*} \ge \beta_n > 0$; otherwise, the source node enqueues the packet for later transmission \cite{Guan-2016-ToTransmit}. Thus, the cumulative distribution function (CDF) of Rician distribution is determined based on the fading threshold $\beta_n$, which is defined as:
\begin{align}
\mathcal{F}_{\Tilde{h}_{n}^f}(\beta_n) & = \int_{0}^{\beta_n} xe^{-\frac{x^2+b^2}{2}}I_0(xb) \,dx
= 1-Q_1(b,\beta_n).
\end{align}
Here, $Q_1(.,.)$ denotes the first-order Marcum $Q$-function. Let $|\boldsymbol{F}|$ be the cardinality of the set $\boldsymbol{F}$, which specifies the number of sub-channels. Accordingly, the transmission probability of a packet from the source node in a time slot over the Rician fading can be expressed as:
\begin{align} \label{murice}
\begin{aligned}
\mu_n(\beta_n) = 1-\mathcal{F}_{\Tilde{h}_n^{f^*}}(\beta_n) = 1-(1-Q_1(b,\beta_n) )^{|\boldsymbol{F}|}. 
\end{aligned}
\end{align}
A similar approach can be applied to the Rayleigh distribution for the NLoS channel by assuming $b = 0$ and the Rayleigh fading factor $\Omega = 2$ \cite{Ghazikor-2023-Exploring}.

\subsection{Queuing Analysis} \label{Queuing Analysis}
\subsubsection{Time Threshold Model}
The source node $n \in \boldsymbol{N}$ might transmit time-sensitive packets, such as C2 packets, to UAV nodes. Hence, it is important to ensure a reliable delivery of these packets to the destination before a predetermined timeout. Let $T_n$ represent the waiting time in the queue for the source node. Sometimes, the source node cannot transmit the packet due to poor channel conditions, such as low SINR; any packet with a waiting time $T_n$ exceeding the time threshold $T_n^{th}$ will be discarded \cite{Xiaoqing-2005-Congestion}.

Let $t_n$ be the number of time slots necessary for the source node to transmit a packet to the destination through the channel. Using exponential distribution and $\mu_n(\beta_n)$ derived from Eq. \eqref{murice}, the PDF of $ t_n$ can be approximated as (refer to \cite{Guan-2016-ToTransmit} for the accuracy of this approximation):
\begin{align} \label{approximated_pdf}
\begin{aligned}    
\mathrm{f}_{t_n} (x) & = (1 - \mu_n(\beta_n))^{x - 1}\mu_n(\beta_n)
\\ & \approx \mu_n(\beta_n)e^{-\mu_n(\beta_n)x}. 
\end{aligned}
\end{align}
Finally, the probability of dropping a packet, denoted as $\mathbb{P}_n^{dly} (\beta_n)$ due to exceeding the maximum queuing delay threshold $T_n^{th}$ can be represented as~\cite{Guan-2016-ToTransmit}:
\begin{align}
\mathbb{P}_n^{dly} (\beta_n) \triangleq \mathbb{P}(T_n > T_n^{th}) = e^{-(\frac{\mu_n(\beta_n)} {T^{slt}_n} - \lambda_n) T_n^{th}}, 
\end{align}
where $\lambda_n$ represents the average incoming packet rate following a Poisson distribution and $T^{slt}_n$ is the time slot duration. 
\noindent
\textbf{Upper Bound of $\beta_n $ and $\lambda_n$.} The upper bound of $\beta_n$ and $\lambda_n$ for Rician distribution can be determined by $\mathbb{P}_n^{dly} (\beta_n) \le 1$, and by solving the following inequality for $\beta_n$:
\begin{align}
Q_1(b,\beta_n) \ge 1-(1-\lambda_n T^{slt}_n)^{\frac{1}{|\boldsymbol{F}|}}. 
\end{align}
From this inequality, we can also derive the following upper bound for $\lambda_n$: 
\begin{align}
\lambda_n \le \frac{1-(1- Q_1(b,\beta_n) )^{|\boldsymbol{F}|}}{T_n^{slt}}.
\label{eq:lambda-upper-bound}
\end{align}
\noindent  We will use this upper bound to characterize the time complexity of our proposed algorithms in Appendix E.

\subsubsection{Buffer Overflow Model}
In addition to the time threshold model that addresses the time sensitivity of data packets, queues are assumed to have finite buffer sizes. In cases where the source node cannot transmit packets due to poor channel conditions, buffer overflow may occur. When the buffer reaches full capacity, any incoming packets will be dropped. Using the principles of queuing theory, the probability of exceeding the buffer capacity in a particular state $i$, where $i \ge 0$, is determined as \cite{Gross-2008-Fundamentals}:
\begin{align}
\overline{\mathbb{P}_{i,i+1}} = \mathbb{P}[x_1+...+x_{i+1} > b_n | x_1+...+x_i \le b_n].
\end{align}
Here, $b_n$ denotes the buffer capacity for the source node $n \in \boldsymbol{N}$, and $x$ is the packet length following an exponential random variable with a parameter $\eta_n$. 
By applying the Markov chain, the probability of buffer overflow, which represents the probability of packet loss, can be approximated as \cite{Ghazikor-2023-Exploring}:
\begin{align*}
\mathbb{P}_n^{ov}(\beta_n) & \approx \sum_{i=0}^\infty \overline{\mathbb{P}_{i,i+1}}\pi_i 
 = \frac{(1-\rho_n(\beta_n))e^{-\Tilde{b}_n(1-\rho_n(\beta_n))}}{1-\rho_n(\beta_n) e^{-\Tilde{b}_n(1-\rho_n(\beta_n))}},
\end{align*}
where $\rho_n(\beta_n) = \frac{\lambda_nT^{slt}_n}{\mu_n(\beta_n)}$ and $\Tilde{b}_n = b_n \eta_n$ are the offered load and normalized buffer capacity, respectively. Notably, the buffer overflow model resembles the M/M/1/K queue model by assuming $\Tilde{b}_n = K \ge 50$. Further details and proof are provided in Appendix~\ref{bufferoverflowmodel}.

\subsection{Interference Analysis} \label{Channel Analysis}
In this section, we evaluate the impact of interference on the main channel between the source node and its destination. We aim to determine the probability of transmission error due to high interference from other nodes operating in the same spectrum band. Assuming that $\gamma_{th}$ denotes the SINR threshold, a transmission error occurs when the SINR falls below the SINR threshold $\gamma_{th}$. Let $I_n^f(\boldsymbol{\beta}_{-n})$ represent the impact of interferer nodes on the destination, which is given by \cite{Guan-2016-ToTransmit}:
\begin{align} \label{interference}
I_n^f(\boldsymbol{\beta}_{-n}) = \sum_{m\in \boldsymbol{N} \backslash n} P_m(\hat{h}_{m}^f)^2(\Tilde{h}_{m}^f)^2\alpha_m^f(\beta_m), 
\end{align}
where $(\hat{h}_{m}^f)^2$ and $(\Tilde{h}_{m}^f)^2$ are the path loss and the square of fading coefficient between the interferer nodes and the destination, respectively. Furthermore, $P_m$ and $\boldsymbol{\beta}_{-n} \triangleq (\boldsymbol{\beta}_{\boldsymbol{m}})_{\boldsymbol{m} \subset \boldsymbol{N} \backslash n}$ are the transmission power and the fading threshold of the interferer nodes. In addition, $\alpha_m^f(\beta_m)$ equals one if interferer node $m$ transmits using channel $f$, and zero otherwise. Thus, the probability of transmission error $\mathbb{P}_n^{err}(\boldsymbol{\beta})$ is defined as:
\begin{align} \label{perr}
\mathbb{P}_n^{err}(\boldsymbol{\beta}) \triangleq \mathbb{P}(\gamma_n < \gamma_{th}) = \mathbb{P}\Big(\frac{P_n(\hat{h}_n^f)^2(\Tilde{h}_n^f)^2}{\sigma^2+I_n^f(\boldsymbol{\beta}_{-n})} < \gamma_{th} \Big),
\end{align}
where $P_n$ denotes the transmission power of the source node, and $\sigma^2  = kTW$ represents the thermal noise power where $k$, $T$, and $W$ are the Boltzmann constant, noise temperature, and bandwidth, respectively.

By adopting a classical stochastic geometry approach, the aggregated interference $I_n^f(\boldsymbol{\beta}_{-n})$ can be modeled using Log-normal distribution, which better fits the real distribution than Gamma distribution \cite{Tian-2016-Interference, Tian-2018-QoS}. Thus, PDF of $I_n^f(\boldsymbol{\beta}_{-n})$ is given by:
\begin{align}
    \zeta_n(x) = \mathrm{f}_{I_n^f(\boldsymbol{\beta}_{-n})} (x) = \frac{\text{exp}(-\frac{(\text{ln}x-\mu(\boldsymbol{\beta}_{-n}))^2}{2\sigma^2(\boldsymbol{\beta}_{-n})})}{x\sigma(\boldsymbol{\beta}_{-n})\sqrt{2\pi}},
\end{align}
where $\mu(\boldsymbol{\beta}_{-n})$ and 
$\sigma(\boldsymbol{\beta}_{-n})$ are the location and scale parameters in Log-normal distribution, respectively, which are derived in Appendix~\ref{log-normaldistribution}. Then, the complementary cumulative distribution function (CCDF) of $I_n^f(\boldsymbol{\beta}_{-n})$ can be expressed as:
\begin{equation}
\begin{aligned}
& v_n(x,\boldsymbol{\beta}_{-n}) \triangleq 1 - \mathcal{F}_{I_n^f(\boldsymbol{\beta}_{-n})} (x) =  \frac{1}{2} - \\
&\frac{1}{2} \text{erf}(\frac{\ln(x)-\mu(\boldsymbol{\beta}_{-n})}{\sqrt{2}\sigma(\boldsymbol{\beta}_{-n})}) = 1-\phi(\frac{\ln(x)-\mu(\boldsymbol{\beta}_{-n})}{\sigma(\boldsymbol{\beta}_{-n})}).
\end{aligned}
\end{equation}
Here, $\text{erf}(x) = \frac{2}{\sqrt{\pi}} \int_0^x e^{-t^2}dt$ is the error function and $\phi(x)$ denotes the CDF of the standard normal distribution. Eventually, $\mathbb{P}_n^{err}(\boldsymbol{\beta})$ using Eq. \eqref{perr} can be defined as \cite{Guan-2016-ToTransmit, Tian-2016-Interference}:
\begin{align*}
\mathbb{P}_n^{err}(\boldsymbol{\beta}) & = \mathbb{P}\Big(\frac{P_n(\hat{h}_n^f)^2(\Tilde{h}_n^f)^2}{\gamma_{th}}-\sigma^2 < I_n^f(\boldsymbol{\beta}_{-n}) \Big) \\ 
& = \int_{\beta_n}^{\infty} \mathrm{f}_{\Tilde{h}_n^f} (x) v_n\Big(\frac{P_n(\hat{h}_n^f)^2}{\gamma_{th}}x^2-\sigma^2,\boldsymbol{\beta}_{-n}\Big)dx,
\end{align*}
where $\mathrm{f}_{\Tilde{h}_n^f} (x)$ represents the PDF of Rician or Rayleigh distributions between the source node and its destination.

\subsection{Distributed Throughput Optimization}
In the queuing and interference analysis sections, we calculated the probability of time threshold $\mathbb{P}_n^{dly}(\beta_n)$ due to exceeding the maximum queuing delay, the probability of buffer overflow $\mathbb{P}_n^{ov}(\beta_n)$ due to full buffer capacity, and the probability of transmission error $\mathbb{P}_n^{err}(\boldsymbol{\beta})$ caused by SINR falling below the SINR threshold where the behavior of them with respect to different parameters are shown in Figs. \ref{Pdly_Tth}, \ref{Pov_Tslt} and \ref{Perr_Nodes}, respectively. Due to the negligible product of $\mathbb{P}_n^{dly}(\beta_n)$, $\mathbb{P}_n^{ov}(\beta_n)$ and $\mathbb{P}_n^{err}(\boldsymbol{\beta})$, the probability of overall loss $\mathbb{P}_n^{lss}(\boldsymbol{\beta})$ can be approximated by:
\begin{align*}
\small
\begin{aligned}
\mathbb{P}_n^{lss}(\boldsymbol{\beta}) & = \mathbb{P}_n^{ov}(\beta_n) + [1-\mathbb{P}_n^{ov}(\beta_n)]\mathbb{P}_n^{dly}(\beta_n) + [1-\mathbb{P}_n^{ov}(\beta_n)] \times \\ & [1-\mathbb{P}_n^{dly}(\beta_n)] \mathbb{P}_n^{err}(\boldsymbol{\beta})
\approx \mathbb{P}_n^{ov}(\beta_n) + \mathbb{P}_n^{dly}(\beta_n) + \mathbb{P}_n^{err}(\boldsymbol{\beta}).
\end{aligned}
\end{align*}

\noindent 
Given the probability of overall loss $\mathbb{P}_n^{lss}(\boldsymbol{\beta})$, the throughput $\mathbb R_n(\boldsymbol{\beta})$ is approximated as:
\begin{align} \label{Rn}
\begin{aligned}
\mathbb R_n(\boldsymbol{\beta}) & = \lambda_n [1-\mathbb{P}_n^{lss}(\boldsymbol{\beta})] 
\\ &
\approx  \lambda_n [1-\mathbb{P}_n^{dly}(\beta_n)-\mathbb{P}_n^{ov}(\beta_n)-\mathbb{P}_n^{err}(\boldsymbol{\beta})]. 
\end{aligned}
\end{align}
In Appendix C, we provide further justifications and numerical results in Fig. \ref{exact_approx} to demonstrate the accuracy of this approximation compared with the exact throughput  $\mathbb R_n(\boldsymbol{\beta})$.

\begin{theorem}
Given the provided expressions, $\mathbb R_n(\boldsymbol{\beta})$ is a concave function with respect to $\beta_n$.
\end{theorem}
\begin{proof}
The proof is provided in Appendix~\ref{concave_beta}.
\end{proof}
With this theorem, we aim to develop a distributed algorithm by which different nodes adjust their fading thresholds $\beta_i$ to maximize the average throughput. In this case, distributed nodes should coordinate to converge to the optimal fading threshold or transmission policy, which is advantageous for all nodes. To solve this problem, we use \emph{consensus-based distributed optimization}, which tries to reach a consensus on the fading threshold $\boldsymbol{\beta}$ as different nodes collaborate~\cite{Berahas-2019-Balancing}. Our proposed Distributed Transmission Control (DTC) policy is presented in Algorithm \ref{alg:DTC} by which we maximize the average throughput by finding the optimal fading threshold $\boldsymbol{\beta}^{\star}$.

\begin{algorithm}[t]
\small
\caption{Distributed Transmission Control (DTC)} \label{alg:DTC}
\begin{algorithmic}[1]
\Function{DTC}{$\boldsymbol{M}$, $\boldsymbol{stp}$, $i_{max}$, $\boldsymbol{\beta}^{max}$}
    \State $ \boldsymbol{\beta} \gets \boldsymbol{\beta}^{max}$
    \For{$n$ \textbf{in} \textbf{range} $\boldsymbol{M}$}
        \State $\mathbb R^{\star} \gets \mathbb R_n(\boldsymbol{\beta})$, $\beta_n \gets \boldsymbol{\beta}[n]$, $\boldsymbol{m} \gets \boldsymbol{M} - \{n\}$
        \State $\boldsymbol{\beta}^{slf}[n] \gets \Call{LCS}{\beta_n, \boldsymbol{stp}, \mathbb R^{\star}, \beta^{max}_n}$
    \EndFor
    \State $\boldsymbol{\beta}^{\star} \gets \boldsymbol{\beta}^{slf}$
    \For{$i$ \textbf{in} \textbf{range} $i_{max}$}
        $\boldsymbol{\beta} \gets \boldsymbol{\beta}^{\star}$
        \For{$n$ \textbf{in} \textbf{range} $\boldsymbol{M}$}
            \State $\beta_n \gets \boldsymbol{\beta}[n]$, $\boldsymbol{m} \gets \boldsymbol{M} - \{n\}$
            \State $\mathbb R^{\star} \gets \textbf{R}[i][n] \gets \mathbb R_n(\boldsymbol{\beta})$
            \State $\boldsymbol{\beta}^{\star}[n] \gets \Call{LCS}{\beta_n, \boldsymbol{stp}, \mathbb R^{\star}, \beta^{max}_n}$
        \EndFor
            \If{$|\boldsymbol{\beta} - \boldsymbol{\beta}^{\star}| < \epsilon$}
            \State \textbf{break}
        \EndIf
    \EndFor
    \State \textbf{return} $\boldsymbol{\beta}^{\star}$
\EndFunction
\end{algorithmic}
\end{algorithm}

Initially, this algorithm sets the fading threshold $\boldsymbol{\beta}$ to the maximum value $\boldsymbol{\beta}^{max}$, which is obtained from the upper bound of the fading threshold. The algorithm then finds the selfish fading threshold $\boldsymbol{\beta}^{slf}$ for all nodes, where each node assumes that the other nodes are not transmitting. Then, it starts to discover the optimal fading threshold $\boldsymbol{\beta}^{\star}$ by the initialized selfish values. To do this, the DTC algorithm utilizes the Local Coordinate Search (LCS) function presented in Appendix~\ref{LCS} to determine the maximum throughput and associated best fading threshold for each node in each iteration. In the end, if the difference in the fading thresholds for two consecutive iterations is less than $\epsilon$, this algorithm returns the optimal fading threshold $\boldsymbol{\beta}^{\star}$. At each iteration, the calculated throughput $\mathbb R_n(\boldsymbol{\beta})$ will be stored in the throughput set $\textbf{R}$, where the first dimension denotes the index of iterations and the second dimension represents the index of the nodes. Furthermore, the time complexity of the proposed DTC algorithm is provided in Appendix \ref{LCS}. In the next section, we extend this analysis to incorporate video transmission control policies. 

\section{Joint Distributed Video Transmission and Encoder Control}
\label{Proposed Solution}
In the previous section, we presented the DTC algorithm, which aims to maximize performance through distributed transmission control. In this section, we further extend our setup by incorporating video encoding rate optimization. Therefore, the performance metric is expanded from focusing solely on throughput optimization to include PSNR optimization, which accounts for video distortion.

\subsection{Video Distortion and PSNR Analysis}
Given the probability of overall loss $\mathbb{P}_n^{lss}(\boldsymbol{\beta})$, we define the packet loss distortion as $\mathrm{D}_n^{lss} (\boldsymbol{\beta}) = s_n\mathbb{P}_{n}^{lss}(\boldsymbol{\beta})$, where $s_n$ is the sensitivity parameter of a video sequence to packet loss. Moreover, we assume $\mathrm{D}_n^{cmp}(\mathrm{E}_n) = D_{0} + \frac{\theta_0}{\mathrm{E}_n - E_{0}}$ represents the lossy video compression distortion. Here, $\mathrm{E}_n = \lambda_n L_n$ denotes the video encoding rate at the application layer, where $L_n$ is the average video packet length. Also, $D_{0}$, $\theta_0$, and $E_{0}$ are the parameters of video rate-distortion model, which can be determined using nonlinear regression techniques \cite{Xiaoqing-2005-Congestion}. According to the packet loss distortion $\mathrm{D}_n^{lss} (\boldsymbol{\beta})$ and the lossy video compression distortion $\mathrm{D}_n^{cmp}(\mathrm{E}_n)$, the overall video distortion $\mathrm{D}_n(\mathrm{E}_n, \boldsymbol{\beta})$ is given by \cite{Tian-2016-Interference}:
\begin{align}
\begin{aligned}
\mathrm{D}_n(\mathrm{E}_n, \boldsymbol{\beta}) & = \mathrm{D}_n^{cmp}(\mathrm{E}_n) + \mathrm{D}_n^{lss} (\boldsymbol{\beta})
\\ &
= D_{0} + \frac{\theta_0}{\mathrm{E}_n - E_{0}} + s_n\mathbb{P}_n^{lss}(\boldsymbol{\beta}).
\end{aligned}
\end{align}
\noindent
Given the overall video distortion, we calculate the PSNR, an important metric for video transmission quality measurement. The PSNR metric is defined by the widely used mean square error rate-distortion model. Therefore, given the overall video distortion $\mathrm{D}_n(\mathrm{E}_n, \boldsymbol{\beta})$, PSNR between the source node $n \in \boldsymbol{N}$ and its destination is calculated as \cite{Tian-2016-Interference, Tang-2021-QoE}:
\begin{align} \label{psnr}
    \mathcal{P}_n(\mathrm{E}_n, \boldsymbol{\beta}) = 10\log_{10} \Big(\frac{(2^p - 1)^2}{\mathrm{D}_n (\mathrm{E}_n, \boldsymbol{\beta})}\Big),
\end{align}
where $p$ is the bit-depth of a video pixel. In Eq. \eqref{psnr}, the encoding rate $\mathrm{E}_n$ and fading threshold $\boldsymbol{\beta}$ are jointly incorporated into our cross-layer formulation including packet loss distortion $\mathrm{D}_n^{lss} (\boldsymbol{\beta})$ at the PHY and MAC layers and lossy video compression distortion $\mathrm{D}_n^{cmp}(\mathrm{E}_n)$ at the application layer.

\begin{algorithm}[t]
\small
\caption{Joint Distributed Video Transmission and Encoder Control (JDVT-EC)} \label{alg:JDVT-EC}
\begin{algorithmic}[1]
\Function{JDVT-EC}{$\boldsymbol{\mathrm{E}}^{ini}$, $\boldsymbol{\beta}^{ini}$, $\boldsymbol{M}$, $\boldsymbol{stp}$, $i_{max}$, $\boldsymbol{\beta}^{max}$, $\boldsymbol{\lambda}^{max}$}
    \State $\boldsymbol{\mathrm{E}}^{\star} \gets \boldsymbol{\mathrm{E}}^{ini}$, $ \boldsymbol{\beta}^{\star} \gets \boldsymbol{\beta}^{ini}$
        \For{$ctr$ \textbf{in} \textbf{range} $i_{max}$} $\boldsymbol{\mathrm{E}}^{old} \gets \boldsymbol{\mathrm{E}}^{\star}$, $\boldsymbol{\beta}^{old} \gets \boldsymbol{\beta}^{\star}$ 
        \State $\boldsymbol{\beta}^{\star} \gets \Call{DVTC}{\boldsymbol{\mathrm{E}}^{\star}, \boldsymbol{\beta}^{\star}, \boldsymbol{M}, \boldsymbol{stp}, i_{max}, ctr, \boldsymbol{\beta}^{max}}$
        \State $\boldsymbol{\mathrm{E}}^{\star} \gets \Call{DVEC}{\boldsymbol{\mathrm{E}}^{\star}, \boldsymbol{\beta}^{\star}, \boldsymbol{M}, \boldsymbol{stp}, ctr, \boldsymbol{\lambda}^{max}}$
        \If{$|\boldsymbol{\beta}^{old} - \boldsymbol{\beta}^{\star}| < \epsilon$ \textbf{and} $|\boldsymbol{\mathrm{E}}^{old} - \boldsymbol{\mathrm{E}}^{\star}| < \epsilon$}
        \State \textbf{break}
        \EndIf
        \EndFor
    \State \textbf{return} $\boldsymbol{\mathrm{E}}^{\star}$, $\boldsymbol{\beta}^{\star}$
\EndFunction
\end{algorithmic}
\end{algorithm}

To characterize the PSNR in a 3D space, we define the \emph{spatial average PSNR} ($\Bar{\mathcal{P}}_n$) between the source node and its destination as the mean PSNR value calculated over discretized distances and elevation angles, as follows \cite{Wang-2021-Packet}:
\begin{align}
    \Bar{\mathcal{P}}_n = \frac{1}{|\textbf{d}||\Theta|}\sum_{j=\theta_{n}^{min}}^{\theta_{n}^{max}} \sum_{i=d_{n}^{min}}^{d_{n}^{max}} \mathcal{P}^{ij}_n(\mathrm{E}_n, \boldsymbol{\beta}),
\end{align}
where $d_n^{min}$ and $d_n^{max}$ are the minimum and maximum distances, which can be included in a discrete distance set \textbf{d} = \{$d_n^{min}, ..., d_n^{max}$\}. Similarly, $\theta_n = \arctan{(\frac{d_n^V}{d_n^H})}$ denotes the elevation angle where $\theta_n^{min}$ and $\theta_n^{max}$ are the minimum and maximum elevation angles given in a discrete elevation angle set $\Theta$ = \{$\theta_n^{min}, ..., \theta_n^{max}$\}.

\begin{theorem}
Given the expressions provided for PSNR in Eq. \eqref{psnr}, $\mathcal{P}_n(\mathrm{E}_n, \boldsymbol{\beta})$ is concave with respect to $E_n$ and $\beta_n$.
\end{theorem}
\begin{proof}
The proofs are provided in Appendices \ref{concave_en} and \ref{concave_beta}. 
\end{proof}
\noindent 
Based on this result, next we propose a distributed algorithm to maximize the average PSNR values.

\subsection{Distributed PSNR Optimization}
The goal is to maximize the average PSNR for streamer nodes in the environment. To do this, two important parameters, including the encoding rate $\mathrm{E}_n$ and the fading threshold $\beta_n$, need to be optimized for each node.   Therefore, the optimization problem is divided into two sub-problems to solve each parameter individually.  To solve the distributed PSNR optimization, we propose \emph{Joint Distributed Video Transmission and Encoder Control (JDVT-EC)}, as described in Algorithm \ref{alg:JDVT-EC}. In each iteration, the JDVT-EC optimization algorithm finds the optimal fading threshold $\boldsymbol{\beta}^{\star}$ and the optimal encoding rate $\boldsymbol{\mathrm{E}}^{\star}$ using the DVTC and DVEC sub-algorithms and compares it with the previous fading threshold $\boldsymbol{\beta}^{old}$ and the previous encoding rate $\boldsymbol{\mathrm{E}}^{old}$. Finally, if the difference between the optimal and previous sets is smaller than $\epsilon$, the algorithm returns the optimal fading threshold $\boldsymbol{\beta}^{\star}$ and the optimal encoding rate $\boldsymbol{\mathrm{E}}^{\star}$. In this algorithm, $ctr$ counts the number of iterations, and $i_{max}$ specifies the maximum number of iterations.

\begin{algorithm}[t]
\small
\caption{Distributed Video Transmission Control (DVTC)} \label{alg:DVTC}
\begin{algorithmic}[1]
\Function{DVTC}{$\boldsymbol{\mathrm{E}}^{\star}$, $\boldsymbol{\beta}^{\star}$, $\boldsymbol{M}$, $\boldsymbol{stp}$, $i_{max}$, $ctr$, $\boldsymbol{\beta}^{max}$}
    \For{$i$ \textbf{in} \textbf{range} $i_{max}$} $\boldsymbol{\beta} \gets \boldsymbol{\beta}^{\star}$
        \For{$n$ \textbf{in} \textbf{range} $\boldsymbol{M}$}
            \State $\mathrm{E}_n \gets \boldsymbol{\mathrm{E}}^{\star}[n]$, $\beta_n \gets \boldsymbol{\beta}[n]$, $\boldsymbol{m} \gets \boldsymbol{M} - \{n\}$
            \State $\mathcal{P}^{\star} \gets \boldsymbol{\mathcal{P}}[2ctr][i][n] \gets \mathcal{P}_n(\mathrm{E}_n, \boldsymbol{\beta})$
            \State $\boldsymbol{\beta}^{\star}[n] \gets \Call{LCS}{\beta_n, \boldsymbol{stp}, \mathcal{P}^{\star}, \beta^{max}_n}$
        \EndFor
            \If{$|\boldsymbol{\beta} - \boldsymbol{\beta}^{\star}| < \epsilon$}
            \State \textbf{break}
        \EndIf
    \EndFor
    \State \textbf{return} $\boldsymbol{\beta}^{\star}$
\EndFunction
\end{algorithmic}
\end{algorithm}

In the \emph{Distributed Video Transmission Control (DVTC)} algorithm, like the DTC algorithm, ground-aerial nodes collaborate to achieve an optimal distributed transmission strategy that is beneficial to all nodes, assuming that each link can be a main link. Accordingly, multiple nodes coordinate to achieve a consensus on the fading threshold set $\boldsymbol{\beta}$ while keeping the encoding rate constant \cite{Berahas-2019-Balancing}. Each node contains its local information, iteratively communicating with adjacent nodes to determine the optimal fading threshold $\boldsymbol{\beta}^{\star}$. Algorithm \ref{alg:DVTC} presents the DVTC policy to determine the optimal fading threshold set $\boldsymbol{\beta}^{\star}$ for all nodes. In each iteration, if the difference between the updated fading threshold set $\boldsymbol{\beta}^{\star}$ and the preceding one $\boldsymbol{\beta}$ exceeds $\epsilon$, nodes share information about their fading thresholds with each other to specify the optimal fading threshold again. In this algorithm, $n$ and set $\boldsymbol{m}$ denote the source node and interferer nodes while set $\boldsymbol{M}$ contains both the source node and interferer nodes. Also, $\boldsymbol{stp}$ set denotes the step parameters, which are used in the LCS function, and $\boldsymbol{\beta}^{max}$ set represents the maximum value of the fading thresholds corresponding to the upper bounds. Moreover, the LCS function computes the best fading threshold $\beta_n$ for each node while having access to $\boldsymbol{\beta}_{m}$, which explores coordinates by step parameters until it identifies the best fading threshold $\beta_n$ associated with the maximum PSNR $\mathcal{P}_n$ \cite{Recht-2019-Choosing}. In each iteration, the computed PSNR $\mathcal{P}_n$ will be stored in the 3-dimensional $\boldsymbol{\mathcal{P}}$ set, in which the $1^{st}$, $2^{nd}$, and $3^{rd}$ dimensions are related to JDVT-EC algorithm counter, iteration number in the specific algorithm, and the source node index, respectively.

\begin{algorithm}[t]
\small
\caption{Distributed Video Encoder Control (DVEC)} \label{alg:DVEC}
\begin{algorithmic}[1]
\Function{DVEC}{$\boldsymbol{\mathrm{E}}^{\star}$, $\boldsymbol{\beta}^{\star}$, $\boldsymbol{M}$, $\boldsymbol{stp}$, $ctr$, $\boldsymbol{\lambda}^{max}$}
    \State $ \boldsymbol{\beta} \gets \boldsymbol{\beta}^{\star}$
        \For{$n$ \textbf{in} \textbf{range} $\boldsymbol{M}$}
            \State $\mathrm{E}_n \gets \boldsymbol{\mathrm{E}}^{\star}[n]$, $\beta_n \gets \boldsymbol{\beta}[n]$, $\boldsymbol{m} \gets \boldsymbol{M} - \{n\}$
            \State $\mathcal{P}^{\star} \gets \boldsymbol{\mathcal{P}}[2ctr+1][0][n] \gets \mathcal{P}_n(\mathrm{E}_n, \boldsymbol{\beta})$
            \State $\boldsymbol{\mathrm{E}}^{\star}[n] \gets \Call{LCS}{\frac{\mathrm{E}_n}{L_n}, \boldsymbol{stp}, \mathcal{P}^{\star}, \lambda^{max}_n} \times L_n$
        \EndFor
    \State \textbf{return} $\boldsymbol{\mathrm{E}}^{\star}$
\EndFunction
\end{algorithmic}
\end{algorithm}

\begin{table}
    \centering
    \caption{Key Simulation Parameters}
    \resizebox{\columnwidth}{!}{%
    \label{tab:sim_parameters}
    \begin{tabular}{lc}
        \toprule
        Definition & Notation \& Value \\
        \midrule
        Communication Area & $100 \times 100 \, \text{m}^2$ \\
        Environmental Parameters & $\zeta = 20$, $v = 3 \times 10^{-4}$, $\mu = 0.5$\\
        Path Loss Exponent & $\alpha_{L}=2$, $\alpha_{N}=3.5$\\
        Reference Distance & $d_0=10$ m \\
        Number of Sub-Channels & $|\boldsymbol{F}|=14$\\
        Rician Factor & $K_{L}=15$, $K_{N}=1$\\
        Time Threshold & $T^{th}_n=80$ ms\\
        Time Slot Duration & $T^{slt}_n=5$ ms\\
        Normalized Buffer Capacity & $\Tilde{b}_n=100$\\
        Transmission Power & $P_n=0.2$ W\\
        SINR Threshold & $\gamma_{th}=10$\\
        Operating Frequency & $f = 2.4$ GHz\\
        Noise Temperature & $T = 290$ K\\
        Bandwidth & $W = 100$ MHz \\
        Boltzmann Constant & $k = 1.38 \times 10^{-23}$ J/K \\
        Sensitivity Parameter & $s_n = 30$ \\
        Average Video Packet Length & $L_n = 3.04$ Kb \\
        Rate-Distortion Parameters & $D_{0} = 1.18$, $E_{0} = 0.67$, $\theta_0 = 858$\\
        Video Pixel Bit-Depth & $p = 8$ \\
        Fading Threshold Step & $\boldsymbol{stp}_{\beta} = \{ 0.5, 0.01 \}$ \\
        Incoming Packet Rate Step & $\boldsymbol{stp}_{\lambda} = \{ 0.5, 1 \}$ \\
        \bottomrule
    \end{tabular}
    }
\end{table}

To solve the second sub-problem, the \emph{Distributed Video Encoder Control (DVEC)} method, as described in Algorithm \ref{alg:DVEC} aims to find the optimal encoding rate set $\boldsymbol{\mathrm{E}}^{\star}$ for streamer nodes using the LCS algorithm in which, $\lambda_n = \frac{\mathrm{E}_n}{L_n}$ is limited by the maximum incoming packet rates $\lambda_n^{max}$ obtained from the upper bound of $\lambda_n$, as stated in Eq. \eqref{eq:lambda-upper-bound}. Unlike the DVTC algorithm, each source node's video encoding rate $\mathrm{E}_n$ does not impact other nodes' video encoding rates. Therefore, $\boldsymbol{\mathrm{E}}^{\star}$ does not need to be solved iteratively. In this algorithm,  each node in $\boldsymbol{M}$ can be a source node $n \in \boldsymbol{N}$, which can find the optimal video encoding rate $\mathrm{E}_n^{\star}$ and store it in the optimal encoding rate set $\boldsymbol{\mathrm{E}}^{\star}$ without changing $\boldsymbol{\beta}$. Finally, the algorithm returns the optimal encoding rates $\boldsymbol{\mathrm{E}}^{\star}$ for the JDVT-EC algorithm. In Appendix \ref{LCS}, we provide the time complexity of the proposed  JDVT-EC algorithm along with its constituent sub-algorithms.

\begin{table*}[t]
    \begin{minipage}{\linewidth}
    \centering
    \caption{Fading Threshold $\beta_i$ for Different Policies}
    \begin{tabular}{S|SSSSSSSSSS} \toprule
        {Policy} & {Streamer UAV} & {Interferer UAV} & {Node 3} & {Node 4} & {Node 5} & {Node 6} & {Node 7} & {Node 8} & {Node 9} & {Node 10} \\ \midrule 
        $\text{Random}$ & \text{3.16} & \text{3.53} & \text{3.01} & \text{1.48} & \text{2.54} & \text{1.56} & \text{2.49} & \text{1.57} & \text{4.06} & \text{4.08} \\ 
        $\text{Aggressive}$  & \text{3.18} & \text{3.35} & \text{2.03} & \text{2.03} & \text{1.44} & \text{1.88} & \text{1.79} & \text{1.46} & \text{3.43} & \text{4.16} \\
        $\text{Selfish}$  & \text{5.06} & \text{5.35} & \text{2.73} & \text{2.97} & \text{2.39} & \text{2.97} & \text{2.85} & \text{2.95} & \text{5.76} & \text{5.14} \\
        $\text{Fixed}$  & \text{4.00} & \text{4.00} & \text{2.00} & \text{2.00} & \text{2.00} & \text{2.00} & \text{2.00} & \text{2.00} & \text{4.00} & \text{4.00} \\ 
        $\textbf{Optimal}$  & \textbf{5.12} & \textbf{5.44} & \textbf{2.82} & \textbf{2.98} & \textbf{2.57} & \textbf{2.98} & \textbf{2.93} & \textbf{2.97} & \textbf{5.78} & \textbf{5.17} \\ 
        $\text{Conservative}$  & \text{5.33} & \text{5.94} & \text{2.93} & \text{3.06} & \text{2.95} & \text{3.07} & \text{3.20} & \text{3.04} & \text{6.07} & \text{5.36} \\ \bottomrule
    \end{tabular}
    \label{tab:3}
    \end{minipage}
\end{table*}

Furthermore, we note that the majority of prior works reviewed in Section \ref{Related Work} are designed based on centralized algorithms, especially for UAV networks. In this paper, we have developed the DTC and JDVT-EC algorithms, which are distributed optimization algorithms as each node optimizes its fading threshold and video encoding rate locally. This in turn enhances the application of our solutions for large-scale networks.

\section{Evaluation Results} \label{Numerical Results}
To evaluate the performance of the proposed algorithms, we consider a network in which $10$ ground-aerial nodes are distributed according to the Poisson point process (PPP) in the environment. The network includes $1$ streamer UAV (node index 1), $1$ interferer UAV (node index $2$), and $8$ ground nodes (nodes indices $3$ to $10$). In this case, $2$ ground nodes (nodes $10$ and $9$) are, respectively, paired with the streamer and interferer UAVs, and the other $6$ ground nodes (nodes $3$ to $8$) establish $3$ G2G pairs. We assume that the streamer and interferer UAVs are streaming videos to their associated ground nodes in the downlink, also 3 ground nodes are streaming to their paired ground nodes. 
On the other hand, 3 G2G pairs and 2 G2A pairs transmit short packets (e.g., C2 messages). The key simulation parameters are summarized in Table~\ref{tab:sim_parameters}.

\begin{figure}  
    \includegraphics[width=\linewidth, trim={70 0 90 0}, clip]{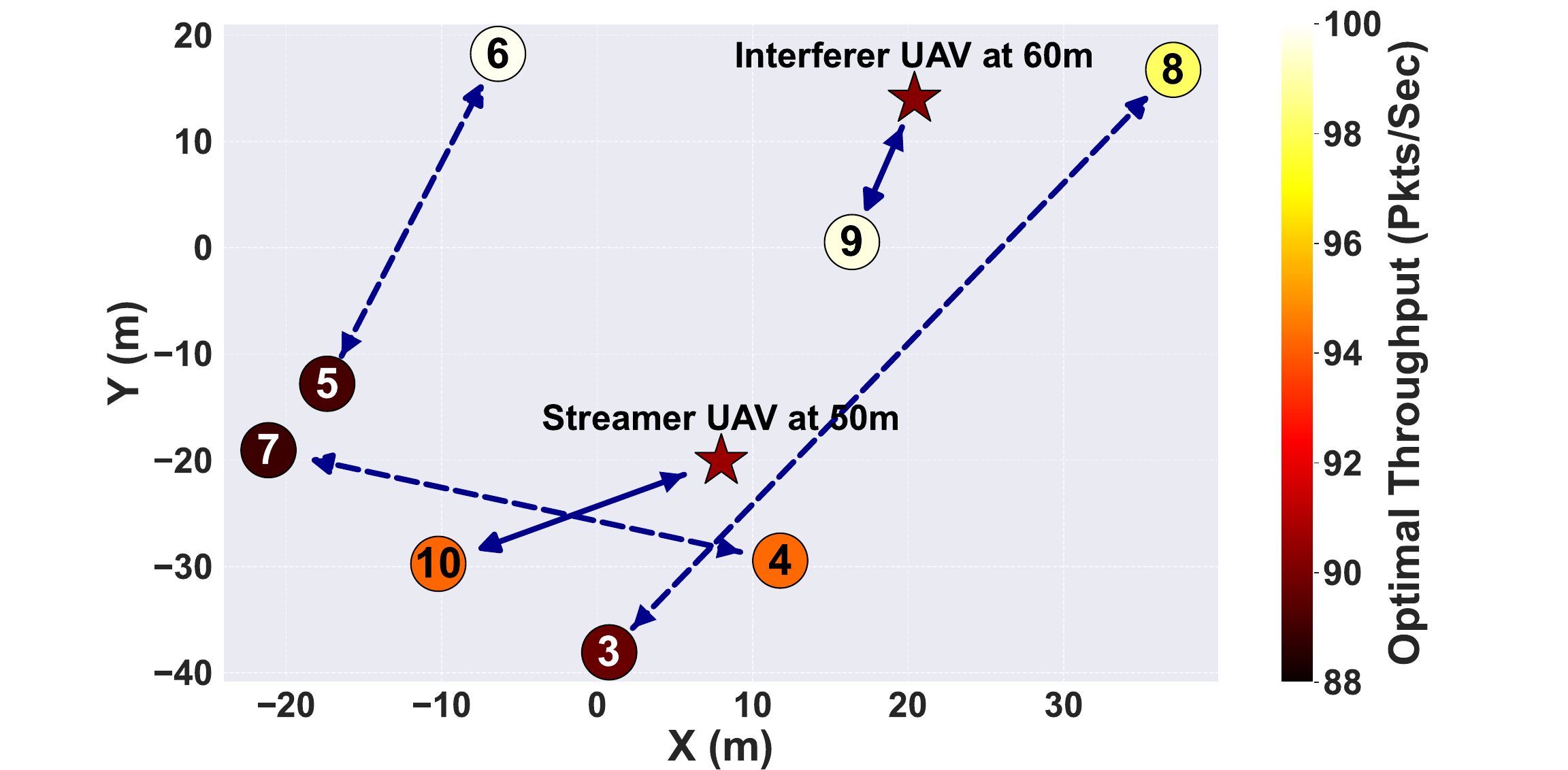}
   \centering
    \caption{Optimal throughput $\mathbb R_i^{\star}$ heatmap achieved by the DTC algorithm. In this simulation, ground and aerial nodes are randomly distributed  according to a PPP, establishing LoS (solid) and NLoS (dashed) links.}
    \label{Y-X}
\end{figure}

\subsection{Throughput Optimization Performance}

\subsubsection{DTC Algorithm Performance}
\begin{table*}[t]
    \centering
    \caption{$\mathbb R_i^{\star} - \beta_i^{\star}$ Values by DTC Algorithm}
    \scriptsize
    \begin{tabular}{S|SSSSSSSSSS} \toprule
        {$\gamma_{th}$} & {Streamer UAV} & {Interferer UAV} & {Node 3} & {Node 4} & {Node 5} & {Node 6} & {Node 7} & {Node 8} & {Node 9} & {Node 10} \\ \midrule 
        \text{5} & \text{96.60 - 5.08} & \text{99.87 - 5.26} & \text{98.88 - 2.77} & \text{88.94 - 2.98} & \text{99.90 - 2.45} & \text{89.46 - 2.97} & \text{95.87 - 2.90} & \text{90.65 - 2.96} & \text{94.40 - 5.73} & \text{93.38 - 5.14} \\ 
        \text{8} & \text{95.01 - 5.11} & \text{99.72 - 5.39} & \text{98.39 - 2.81} & \text{88.93 - 2.98} & \text{99.84 - 2.54} & \text{89.17 - 2.98} & \text{94.82 - 2.92} & \text{89.98 - 2.97} & \text{91.27 - 5.77} & \text{91.33 - 5.16} \\
        \text{10}  & \text{94.20 - 5.12} & \text{99.60 - 5.44} & \text{98.08 - 2.82} & \text{88.92 - 2.98} & \text{99.80 - 2.57} & \text{89.09 - 2.98} & \text{94.25 - 2.93} & \text{89.73 - 2.97} & \text{90.24 - 5.78} & \text{90.60 - 5.17} \\
        \text{12}  & \text{93.58 - 5.13} & \text{99.47 - 5.48} & \text{97.81 - 2.84} & \text{88.92 - 2.98} & \text{99.75 - 2.61} & \text{89.04 - 2.98} & \text{93.88 - 2.93} & \text{89.57 - 2.97} & \text{89.68 - 5.78} & \text{90.14 - 5.17} \\ 
        \text{15}  & \text{92.81 - 5.14} & \text{99.25 - 5.52} & \text{97.44 - 2.85} & \text{88.92 - 2.98} & \text{99.67 - 2.64} & \text{88.99 - 2.98} & \text{93.33 - 2.94} & \text{89.40 - 2.97} & \text{89.27 - 5.78} & \text{89.69 - 5.17}
    \end{tabular}
    \label{tab:2}
    \begin{tabular}{S|SSSSSSSSSS} \midrule
        {$|\boldsymbol{F}|$} & {Streamer UAV} & {Interferer UAV} & {Node 3} & {Node 4} & {Node 5} & {Node 6} & {Node 7} & {Node 8} & {Node 9} & {Node 10} \\ \midrule 
        \text{8}  & \text{85.42 - 4.90} & \text{97.64 - 5.32} & \text{94.20 - 2.62} & \text{82.77 - 2.72} & \text{99.03 - 2.45} & \text{82.79 - 2.72} & \text{86.95 - 2.70} & \text{82.92 - 2.72} & \text{82.84 - 5.52} & \text{83.44 - 4.91} \\ 
        \text{11}  & \text{90.99 - 5.03} & \text{99.15 - 5.40} & \text{96.85 - 2.74} & \text{86.58 - 2.87} & \text{99.61 - 2.53} & \text{86.64 - 2.87} & \text{91.65 - 2.83} & \text{87.01 - 2.86} & \text{87.07 - 5.67} & \text{87.78 - 5.06} \\
        \text{14}  & \text{94.20 - 5.12} & \text{99.60 - 5.44} & \text{98.08 - 2.82} & \text{88.92 - 2.98} & \text{99.80 - 2.57} & \text{89.09 - 2.98} & \text{94.25 - 2.93} & \text{89.73 - 2.97} & \text{90.24 - 5.78} & \text{90.60 - 5.17} \\
        \text{17}  & \text{96.11 - 5.19} & \text{99.78 - 5.47} & \text{98.73 - 2.89} & \text{90.53 - 3.06} & \text{99.87 - 2.61} & \text{90.83 - 3.06} & \text{95.89 - 3.00} & \text{91.67 - 3.05} & \text{92.75 - 5.85} & \text{92.64 - 5.25} \\ 
        \text{20}  & \text{97.25 - 5.24} & \text{99.86 - 5.50} & \text{99.11 - 2.93} & \text{91.71 - 3.13} & \text{99.91 - 2.64} & \text{92.12 - 3.13} & \text{96.87 - 3.05} & \text{93.13 - 3.11} & \text{94.65 - 5.91} & \text{94.22 - 5.31} \\ \bottomrule
    \end{tabular}
\end{table*}

Fig. \ref{Y-X} demonstrates the distribution of nodes in the 100$\times$100 $\text{m}^2$ area in which the star symbol marks the streamer and interferer UAVs at the altitudes of $50$ m and $60$ m, respectively; also, the circles denote the ground nodes. Double arrow lines display the main links between the nodes; the solid ones are the G2A or A2G LoS channels, and the dashed ones are the G2G NLoS channels. Using the heatmap, the optimal throughput $\mathbb R_i^{\star}$ obtained by the DTC algorithm is shown as source nodes transmit their packets to destination nodes. The results show that the destination nodes further away from the ``crowded area''  (such as nodes $6$, $8$, and $9$) experience a better throughput after reaching a consensus on the fading threshold $\beta_i$. However, the streamer and interferer UAVs suffer from A2A interference, resulting in a lower throughput than G2G communication pairs. It should be noted that the nodes are randomly distributed according to a defined PPP, and the proposed DTC algorithm optimizes their fading thresholds.

\begin{figure} [t]
\includegraphics[width=\linewidth, trim={20 0 30 0}, clip]{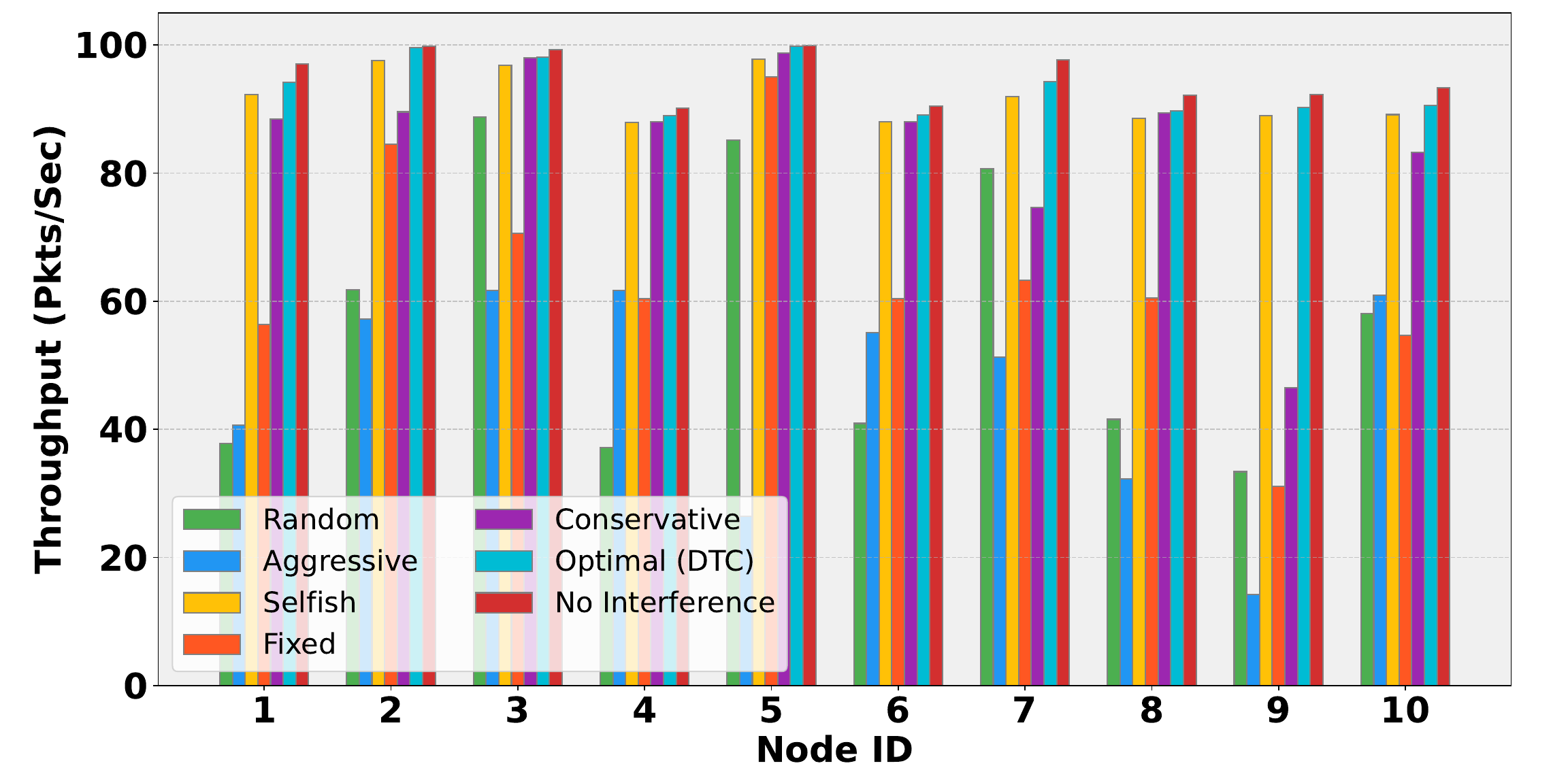}
    \centering
    \caption{$\mathbb R_i$ vs. Different Transmission Policies}
    \label{Exth-NodeID(Policy)}
\end{figure}

In Table~\ref{tab:3}, we compare the DTC algorithm performance with $5$ baselines to set the fading thresholds $\beta_i$, which are defined as: \emph{(1) Random Policy:} Nodes set the fading threshold $\beta_i$ randomly between zero and the upper bound. \emph{(2) Aggressive Policy:} Nodes set the fading threshold $\beta_i$ to low values to send their packets even under bad channel conditions and enqueue fewer packets. Thus, nodes experience more packet losses due to the channel conditions. \emph{(3) Selfish Policy:} Nodes set the fading threshold $\beta_i$ selfishly by assuming that the fading threshold of other nodes is equal to the upper bound (lines 3-6 in the DTC algorithm). Thus, nodes do not collaborate to find the optimal fading threshold $\beta_i^{\star}$ by consensus. \emph{(4) Fixed Policy:} Nodes set the fading threshold $\beta_i$ as a fixed value, which is assumed to be $\beta_i = 4$ for LoS channels and $\beta_i = 2$ for NLoS channels. \emph{(5) Conservative Policy:} Nodes set the fading threshold $\beta_i$ close to the upper bound, so they enqueue more packets instead of transmitting them over the channel. Therefore, nodes experience less packet loss in the channel, but more in the queue.
From the results in Table \ref{tab:3}, we observe that the DTC algorithm finds the optimal fading threshold $\beta_i^{\star}$ that is smaller than the conservative policy, but larger than the aggressive policy for all nodes since the algorithm balances the trade-offs between different packet loss probabilities. In addition, the optimal fading threshold $\beta_i^{\star}$ is slightly higher than the selfish policy since each node increases its fading threshold during collaboration to improve overall system performance. 
Similarly, Fig. \ref{Exth-NodeID(Policy)} compares optimal throughput $\mathbb R_i^{\star}$ with the above baselines, as well as the  ``no-interference'' scenario. The results show that the DTC algorithm outperforms all $5$ policies and improves the average throughput by 1.7\% to 51.65\% while achieving a close performance to the no-interference scenario, which denotes the maximum achievable throughput $\mathbb R_i$ in the environment. Thus, our proposed DTC algorithm balances the trade-off between packet loss due to poor channel condition (i.e., smaller $\beta_i$) vs. packet loss due to buffer overflow and time threshold model (i.e., larger $\beta_i$).

\begin{table*}[t]
    \centering
    \caption{$\mathcal{P}_i - \beta_i - E_i$ Comparison for Individual or Joint Optimized Parameters}
    \begin{tabular}{S|SSSSS|S} \toprule
        {Optimized} & {Streamer UAV} & {Interferer UAV} & {Node 3} & {Node 4} & {Node 5} & {Average $\mathcal{P}^{\star}$} \\ \midrule 
        $\text{$E_i^{\star}$}$ & \text{39.40 - 5.00 - 346.56} & \text{42.12 - 5.00 - 413.44} & \text{38.68 - 2.00 - 422.56} & \text{36.31 - 2.00 - 431.68} & \text{41.71 - 2.00 - 416.48} & \text{39.64 dB} \\ 
        $\text{$\beta_i^{\star}$}$  & \textbf{40.53 - 5.12 - 304.00} & \text{41.97 - 5.44 - 304.00} & \text{41.52 - 2.82 - 304.00} & \text{39.47 - 2.98 - 304.00} & \text{42.03 - 2.57 - 304.00} & \text{41.10 dB} \\
        $\text{$\beta_i^{\star}, E_i^{\star}$}$  & \text{40.37 - 5.12 - 310.08} & \textbf{42.51 - 5.14 - 407.36} & \textbf{41.70 - 2.67 - 373.92} & \textbf{39.50 - 3.02 - 279.68} & \textbf{42.62 - 2.29 - 410.40} & \text{\textbf{41.34} dB} \\ \bottomrule
    \end{tabular}
    \label{tab:6}
\end{table*}

\begin{table*}[t]
    \centering
    \caption{$\mathcal{P}_i - E_i$ Values for Different Encoding Rates}
    \begin{tabular}{S|SSSSS|S} \toprule
        {Encoding Rate} & {Streamer UAV} & {Interferer UAV} & {Node 3} & {Node 4} & {Node 5} & {Average $\mathcal{P}$} \\ \midrule 
        $\text{Low}$ & \text{38.93 - 158.08} & \text{40.67 - 203.68} & \text{38.90 - 179.36} & \text{38.92 - 191.52} & \text{40.05 - 167.20} & \text{39.49 dB} \\ 
        $\text{Medium}$  & \text{40.35 - 325.28} & \text{41.86 - 297.92} & \text{41.42 - 307.04} & \text{39.11 - 331.36} & \text{41.68 - 273.60} & \text{40.88 dB} \\
        $\textbf{Optimal}$  & \textbf{40.37 - 310.08} & \textbf{42.51 - 407.36} & \textbf{41.70 - 373.92} & \textbf{39.50 - 279.68} & \textbf{42.62 - 410.40} & \text{\textbf{41.34} dB} \\
        $\text{High}$  & \text{37.00 - 431.68} & \text{42.34 - 443.84} & \text{41.30 - 425.60} & \text{35.66 - 407.36} & \text{42.42 - 449.92} & \text{39.74 dB} \\ \bottomrule
    \end{tabular}
    \label{tab:5}
\end{table*}

\subsubsection{DTC Algorithm Behavior}
Table~\ref{tab:2} shows the optimal throughput $\mathbb R_i^{\star}$ with different SINR thresholds $\gamma_{th}$ for all nodes. By increasing the SINR threshold $\gamma_{th}$, the optimal throughput $\mathbb R_i^{\star}$ decreases since packets cannot be decoded at the destination with low SINR. Moreover, the optimal fading thresholds $\beta_i^{\star}$ are expressed as a function of different SINR thresholds $\gamma_{th}$. As the SINR threshold $\gamma_{th}$ increases, all nodes increase their optimal fading thresholds $\beta_i^{\star}$ generated by the DTC algorithm since the destination node cannot decode the packets and the source node would prefer to enqueue more packets instead of transmitting to the destination. As mentioned, nodes 3 to 8 establish G2G channels, so their fading threshold ranges are lower than the G2A or A2G channels. Table~\ref{tab:2} also reports the optimal throughput $\mathbb R_i^{\star}$ as the number of sub-channels $|\boldsymbol{F}|$ changes. Increasing the number of sub-channels $|\boldsymbol{F}|$ increases the optimal throughput $\mathbb R_i^{\star}$ as the nodes have more sub-channels available to select. Furthermore, the optimal fading threshold $\beta_i^{\star}$ is demonstrated as the number of sub-channels $|\boldsymbol{F}|$ changes. As the number of sub-channels $|\boldsymbol{F}|$ increases, the optimal fading threshold $\beta_i^{\star}$ increases. This is due to the fact that the upper bound (feasible range) of the fading threshold increases by increasing the number of sub-channels $|\boldsymbol{F}|$. 

\begin{figure}
   \includegraphics[width=\linewidth, trim={40 0 100 60}, clip]{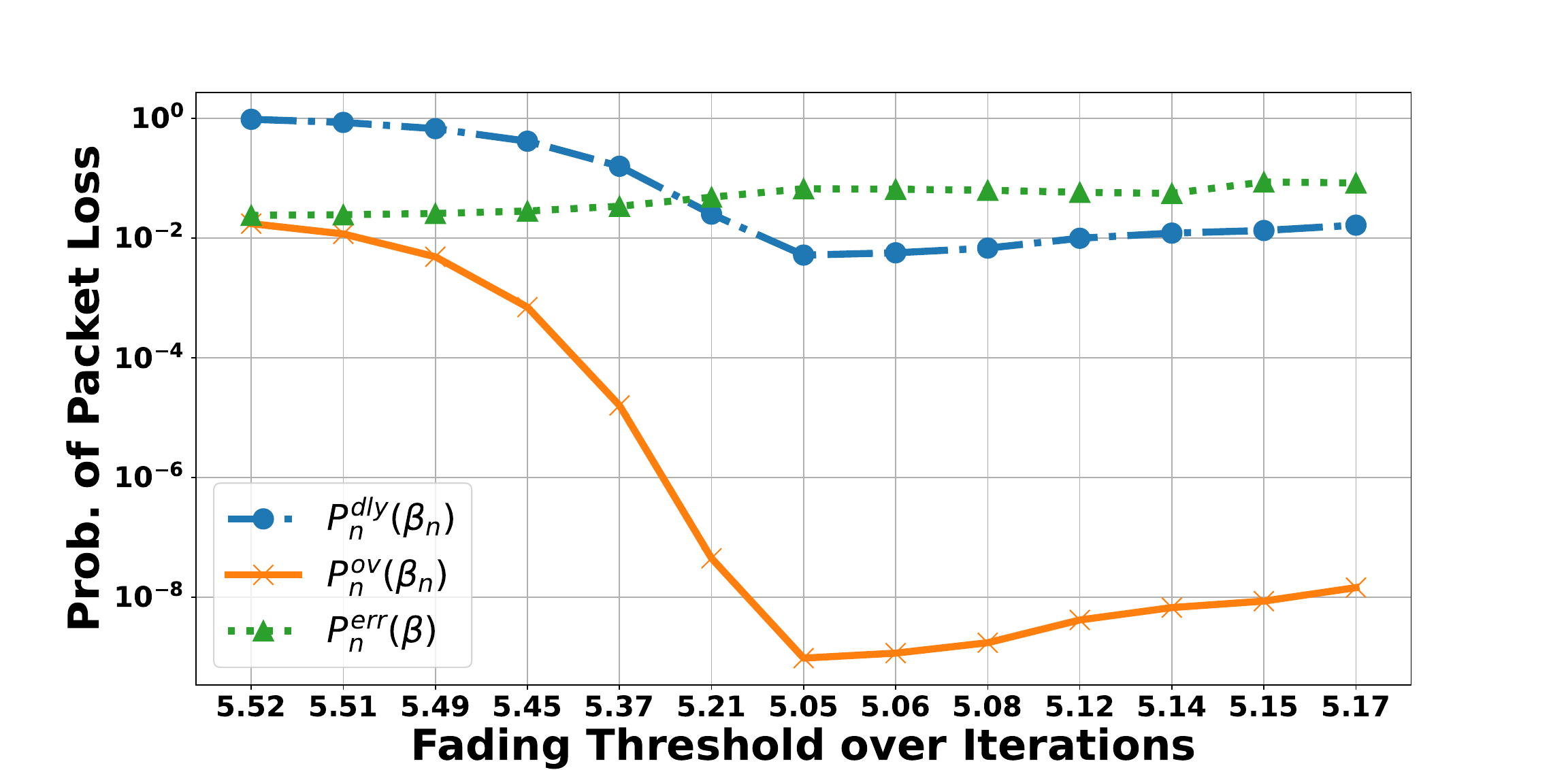}
   \centering
    \caption{Probability of Packet Loss vs. $\beta_n$ for Node 10}
    \label{Plss(node10)-Beta}
\end{figure}

\begin{figure}
   \includegraphics[width=\linewidth, trim={70 0 90 0}, clip]{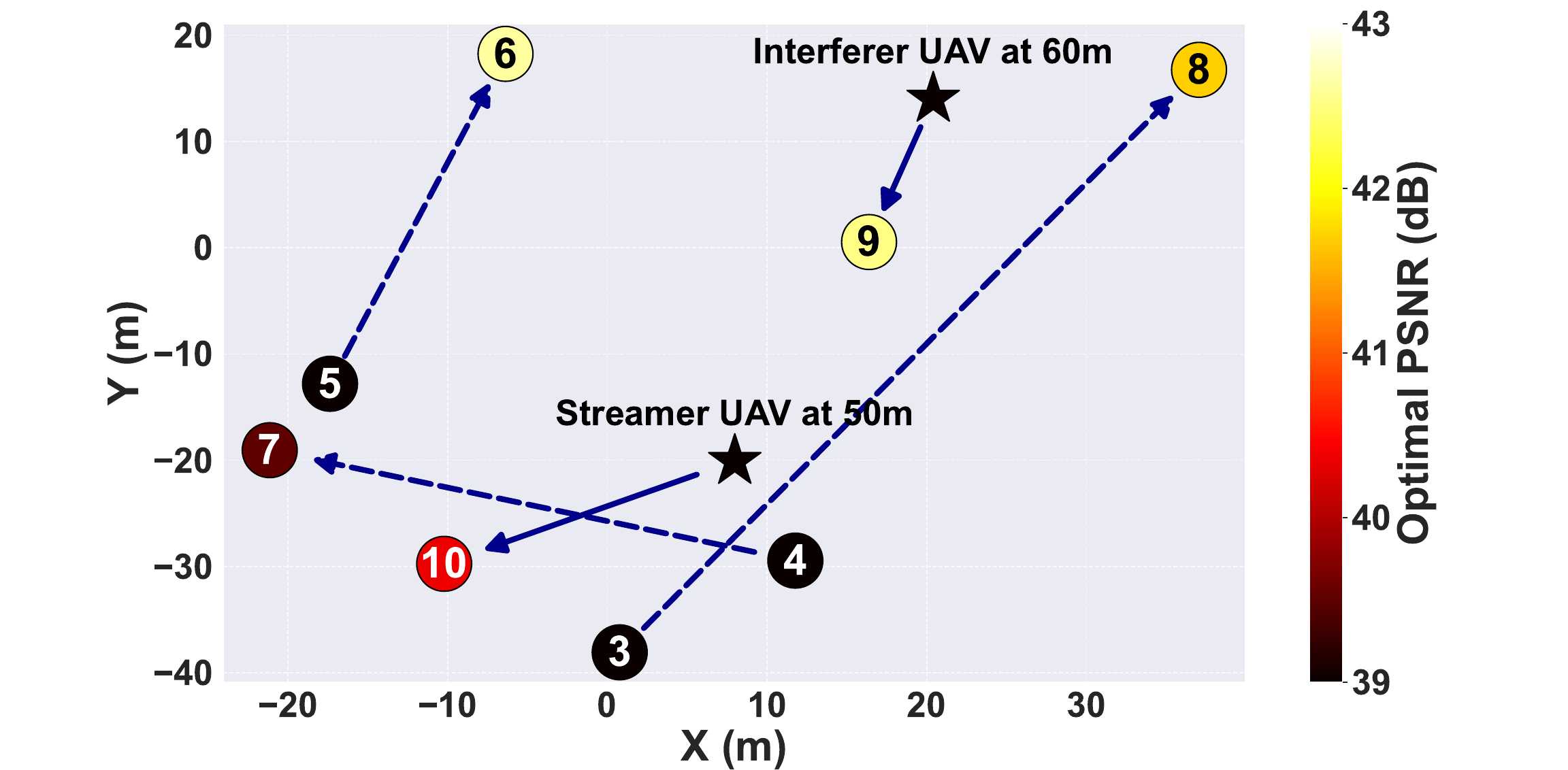}
   \centering
    \caption{Optimal PSNR $\mathcal{P}_i^{\star}$ heatmap achieved by the JDVT-EC algorithm. Ground and aerial nodes are randomly distributed according to a PPP, establishing LoS (solid) and NLoS (dashed) links.}
    \label{PSNR_Dis}
\end{figure}

In Fig. \ref{Plss(node10)-Beta}, the probability of packet loss including $\mathbb{P}_n^{dly}$, $\mathbb{P}_n^{ov}$ and $\mathbb{P}_n^{err}$ are represented for different fading thresholds $\beta_n$ (last one is the optimal fading threshold $\beta_n^{\star}$) when node 10 hits the higher throughput as the DTC algorithm iterates. At the beginning of iterations, node 10 experiences a huge packet loss due to the time threshold $\mathbb{P}_n^{dly}$, and the fading threshold $\beta_n$ is closer to the upper bound. As illustrated, when the fading threshold $\beta_n$ decreases, packet losses in the queue and channel decrease and increase, respectively. However, there are some anomaly points, such as $\beta_n = 5.14$, when node $10$ communicates with other nodes and receives the optimal fading threshold set $\boldsymbol{\beta}^{\star}$.

\subsection{PSNR Optimization Performance}

\subsubsection{JDVT-EC Algorithm Performance}
\begin{table*}
    \centering
    \caption{$\beta_i^{\star} - E_i^{\star}$ Values by JDVT-EC Algorithm}
    \scriptsize
    \begin{tabular}{S|SSSSS|SSSSS} \toprule
        {Node} & {\text{$s_n$ = 20}} & {\text{$s_n$ = 30}} & {\text{$s_n$ = 40}} & {\text{$s_n$ = 50}} & {\text{$s_n$ = 60}} & {\text{$\gamma_{th}$ = 5}} & {\text{$\gamma_{th}$ = 8}} & {\text{$\gamma_{th}$ = 10}} & {\text{$\gamma_{th}$ = 12}} & {\text{$\gamma_{th}$ = 15}} \\ \midrule 
        \text{1} & \text{5.07-337.44} & \text{5.12-310.08} & \text{5.16-288.80} & \text{5.19-273.60} & \text{5.21-261.44} & \text{5.03-334.40} & \text{5.10-316.16} & \text{5.12-310.08} & \text{5.14-304.00} & \text{5.16-297.92} \\
        \text{2}  & \text{5.10-422.56} & \text{5.14-407.36} & \text{5.17-395.20} & \text{5.20-386.08} & \text{5.21-380.00} & \text{4.92-413.44} & \text{5.07-410.40} & \text{5.14-407.36} & \text{5.19-404.32} & \text{5.26-398.24} \\
        \text{3}  & \text{2.62-395.20} & \text{2.67-373.92} & \text{2.72-352.64} & \text{2.74-343.52} & \text{2.77-328.32} & \text{2.62-380.00} & \text{2.64-376.96} & \text{2.67-373.92} & \text{2.70-367.84} & \text{2.73-361.76} \\
        \text{4}  & \text{2.97-310.08} & \text{3.02-279.68} & \text{3.06-258.40} & \text{3.10-237.12} & \text{3.13-221.92} & \text{3.02-279.68} & \text{3.02-279.68} & \text{3.02-279.68} & \text{3.02-279.68} & \text{3.02-279.68} \\
        \text{5}  & \text{2.26-425.60} & \text{2.29-410.40} & \text{2.31-401.28} & \text{2.34-392.16} & \text{2.35-386.08} & \text{2.08-413.44} & \text{2.26-410.40} & \text{2.29-410.40} & \text{2.32-410.40} & \text{2.37-407.36} \\ \bottomrule
    \end{tabular}
    \label{tab:4}
\end{table*}

Fig. \ref{PSNR_Dis} depicts the distribution of all nodes in the environment. The results show the optimal PSNR $\mathcal{P}_i^{\star}$ obtained from the JDVT-EC algorithm for each streamer node. As mentioned, the dashed and solid arrow lines are the main NLoS and LoS links, respectively. In addition, streamer nodes (nodes 1 to 5) are shown by black circles and stars in the environment. It can be seen that nodes 7 and 10 experience a lower optimal PSNR $\mathcal{P}_i^{\star}$ due to being in a ``crowded area'' with high interference. Similar to Fig. \ref{Y-X}, we highlight that the nodes are randomly distributed according to a defined PPP, and the proposed JDVT-EC algorithm optimizes their fading thresholds and video encoding rates.

Table~\ref{tab:6} compares PSNR $\mathcal{P}_i$, fading threshold $\beta_i$, and encoding rate $E_i$ for individual ($E_i^{\star}$ or $\beta_i^{\star}$) or joint ($\beta_i^{\star}, E_i^{\star}$) optimized parameters. In individual optimization, the fading threshold $\beta_i$ and encoding rate $E_i$ are set to 5 and 2 for LoS and NLoS channels and 304 Kbps for streamer nodes, respectively. As mentioned, our objective is to maximize the average $\mathcal{P}^{\star}$ in the joint optimization (JDVT-EC algorithm), not PSNR $\mathcal{P}_i$ for the individual nodes. Accordingly, the JDVT-EC algorithm outperforms the individual optimization in terms of the average $\mathcal{P}^{\star}$ and achieves up to 1.7 dB improvement. However, in terms of individual performance, streamer UAV experiences a higher optimal PSNR $\mathcal{P}^{\star}_i$ in the individual optimized $\beta_i$ compared to the JDVT-EC algorithm. Moreover, the values of PSNR $\mathcal{P}_i$ by different encoding rates $E_i$ with optimal fading threshold $\beta_i^{\star}$ are demonstrated for streamer nodes in Table~\ref{tab:5}. We define different encoding rates according to the tolerance of the transmission queue as follows: \emph{(1) Low Encoding Rate:} Nodes set the encoding rates $E_i$ randomly between 152 Kbps to 212.8 Kbps. \emph{(2) Medium Encoding Rate:} Nodes set the encoding rates $E_i$ randomly between 273.6 Kbps to 334.4 Kbps. \emph{(3) High Encoding Rate:} Nodes set the encoding rates $E_i$ randomly between 395.2 Kbps to 456 Kbps. From the results in Table~\ref{tab:5}, we note that the optimal PSNR $\mathcal{P}_i^{\star}$ achieved by the JDVT-EC algorithm (which uses the optimal encoding rate $E_i^{\star}$) outperforms the baselines with low, medium, and high encoding rates $E_i$ and improves the average PSNR $\mathcal{P}$ up to 1.85 dB.

\subsubsection{JDVT-EC Algorithm Behavior}
In Table~\ref{tab:4}, the optimal fading threshold $\beta_i^{\star}$ for 2 LoS (nodes 1 and 2: streamer and interferer UAVs) and 3 NLoS channels (nodes 3, 4, and 5: ground streamers) is represented by different sensitivity parameters $s_n$. The sensitivity parameter $s_n$ controls the sensitivity of the overall distortion $\mathrm{D}_n$ to the packet loss distortion $\mathrm{D}_n^{lss}$, which balances the weight between the packet loss distortion $\mathrm{D}_n^{lss}$ and the lossy video compression distortion $\mathrm{D}_n^{cmp}$. As shown, increasing the sensitivity parameter $s_n$ increases the optimal fading threshold $\beta_i^{\star}$ since the impact of packet loss distortion $\mathrm{D}_n^{lss}$ on the overall distortion $\mathrm{D}_n$ increases. Therefore, streamer nodes try to act more conservatively in terms of packet loss. Moreover, the optimal encoding rate $E_i^{\star}$ is reported as a function of different sensitivity parameters $s_n$. We note that increasing the sensitivity parameter $s_n$ decreases the optimal encoding rate $E_i^{\star}$ since the JDVT-EC algorithm tries to offset the impact of increasing packet loss distortion $\mathrm{D}_n^{lss}$ by decreasing the optimal encoding rate $E_i^{\star}$ to have a lower probability of overall loss $\mathbb{P}_n^{lss}$. Similarly, Table~\ref{tab:4} shows the optimal fading threshold $\beta_i^{\star}$ by different SINR thresholds $\gamma_{th}$ for streamer nodes. By increasing the SINR threshold $\gamma_{th}$, 
streamer nodes behave more conservatively and increase the optimal fading threshold $\beta_i^{\star}$ to reduce the impact of transmission errors. Therefore, they attempt to decrease the drop in the optimal PSNR $\mathcal{P}_i^{\star}$. In addition, the optimal encoding rate $E_i^{\star}$ for different SINR thresholds $\gamma_{th}$ is reported. From the results, we note that as the SINR threshold $\gamma_{th}$ increases, streamer nodes decrease their optimal encoding rate $E_i^{\star}$ to decrease packet loss. They try to decrease the drop in the optimal PSNR $\mathcal{P}_i^{\star}$ by reducing network load. However, modifying the SINR threshold $\gamma_{th}$ has minimum impact on the optimal encoding rate and fading threshold of node 4.

\begin{figure*}[t]
    \centering
    \begin{minipage}{.33\linewidth}
        \includegraphics[width=6cm, trim={70 10 20 10}, clip]{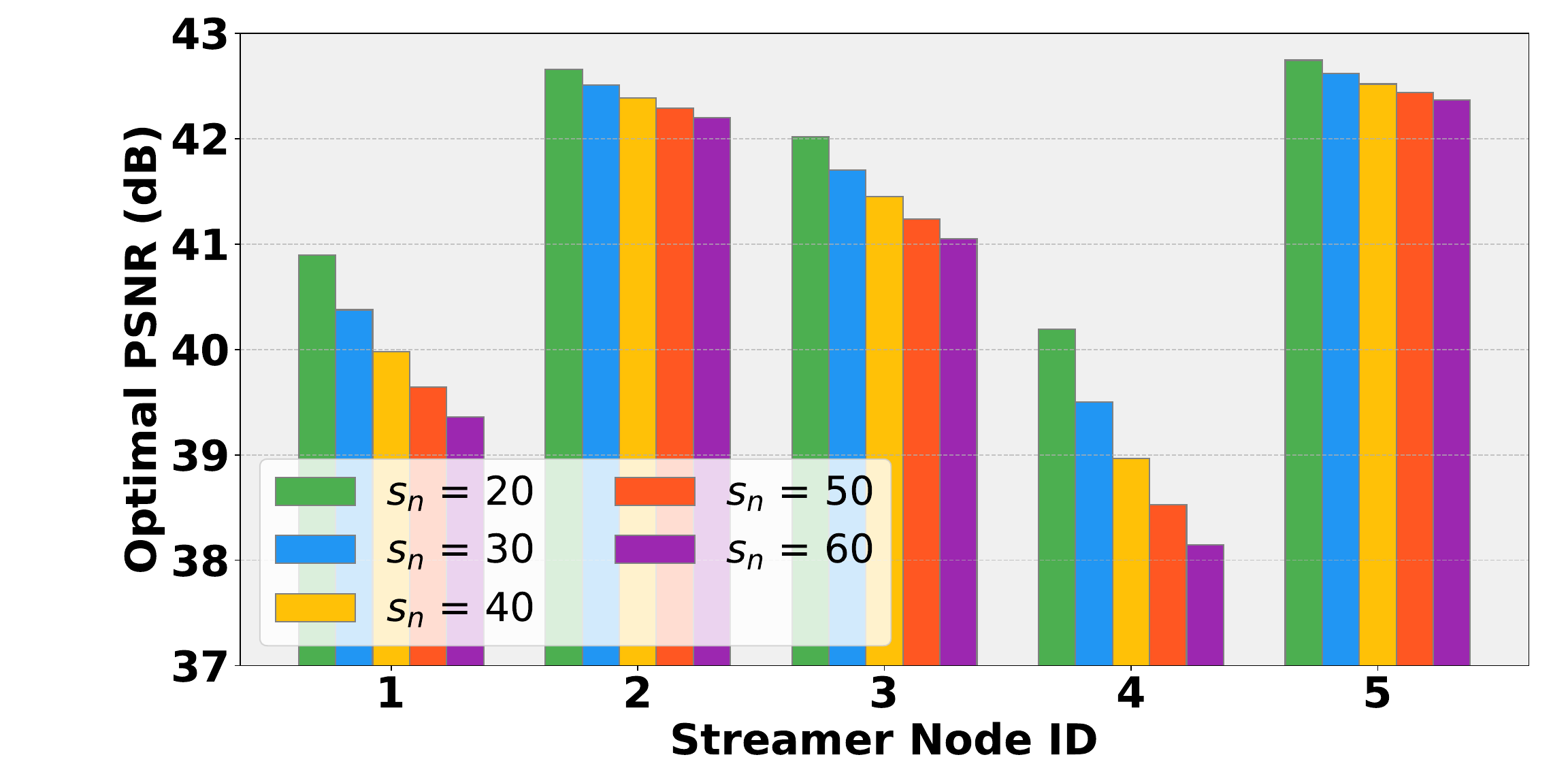}
        \caption{$\mathcal{P}_i^{\star}$ vs. $s_n$ by JDVT-EC}
        \label{PSNR_Sen}
    \end{minipage}%
    \begin{minipage}{.33\linewidth}
        \includegraphics[width=6cm, trim={70 10 20 10}, clip]{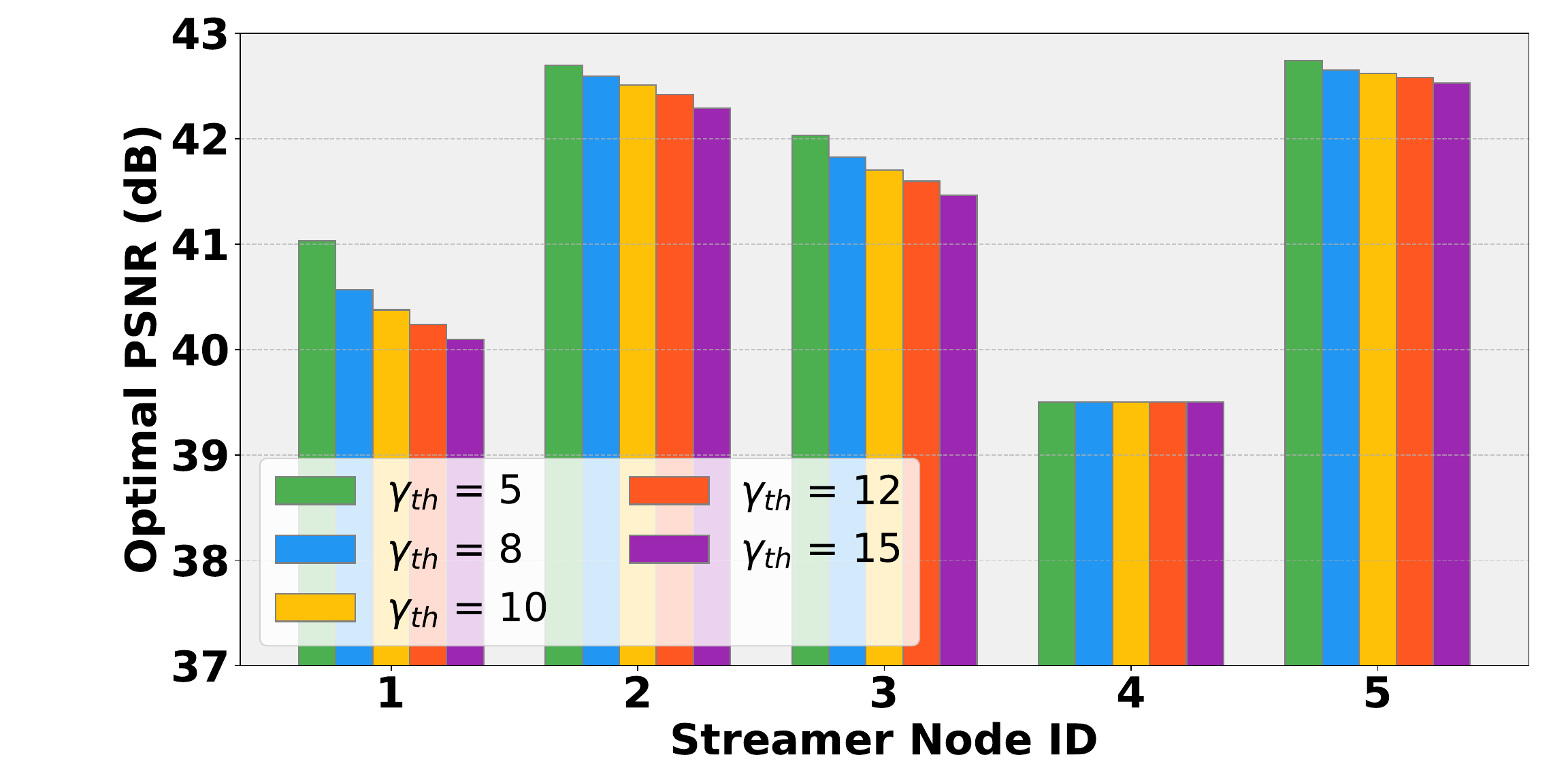}
        \caption{$\mathcal{P}_i^{\star}$ vs. $\gamma_{th}$ by JDVT-EC}
        \label{PSNR_SINRth}
    \end{minipage}%
    \begin{minipage}{.33\linewidth}
        \includegraphics[width=6cm, trim={10 0 30 0}, clip]{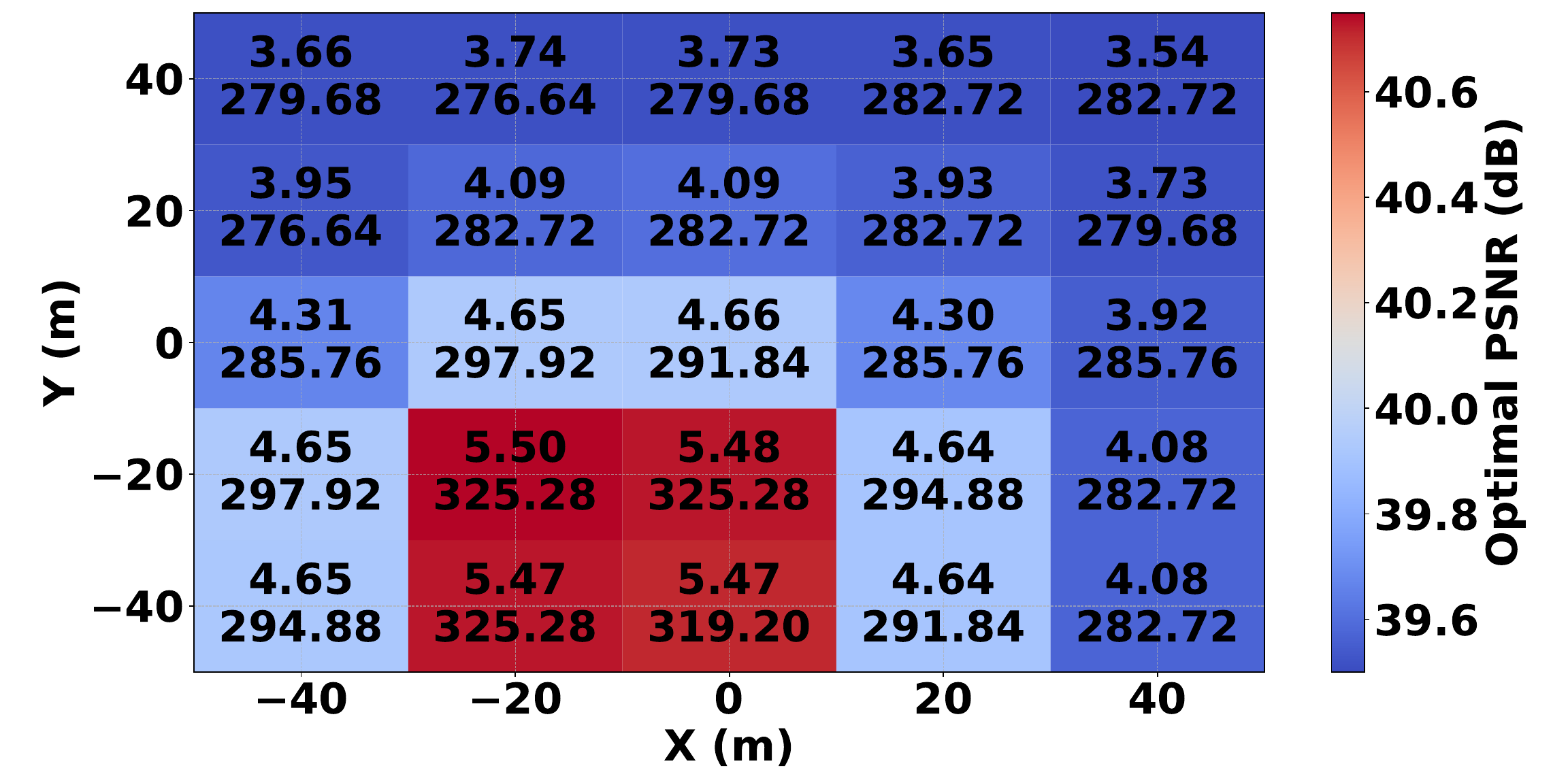}
        \caption{$\mathcal{P}_n^{\star}$ Heatmap for Streamer UAV}
        \label{PSNR_Heatmap}
    \end{minipage}%
\end{figure*}

\begin{figure*}[t]
    \centering
    \begin{minipage}{.33\linewidth}
        \includegraphics[width=6cm, trim={90 10 10 10}, clip]{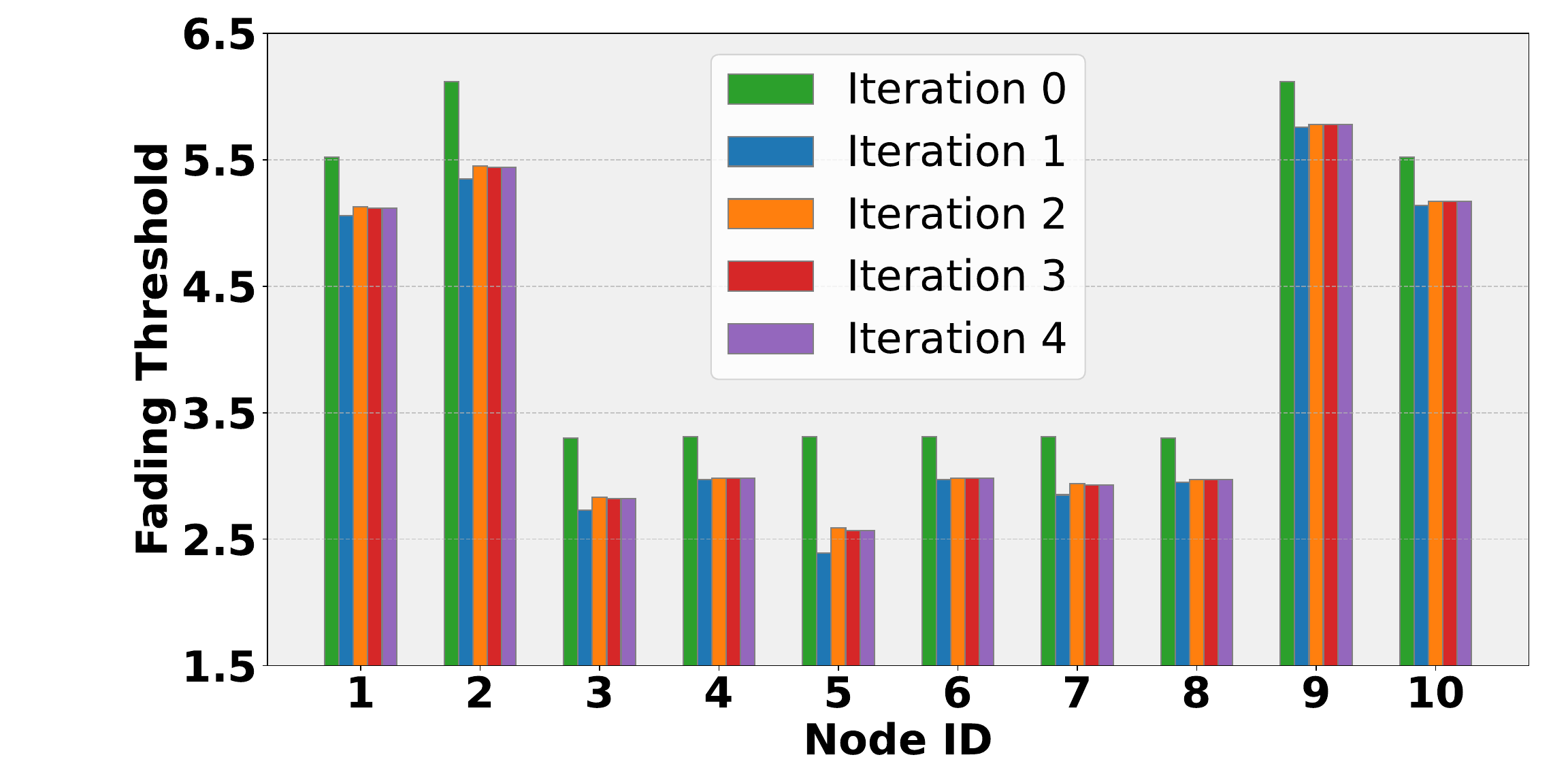}
        \caption{DTC algorithm convergence for $\boldsymbol{\beta}$.}
        \label{dtc_perf}
    \end{minipage}%
    \begin{minipage}{.33\linewidth}
        \includegraphics[width=6cm, trim={90 10 10 10}, clip]{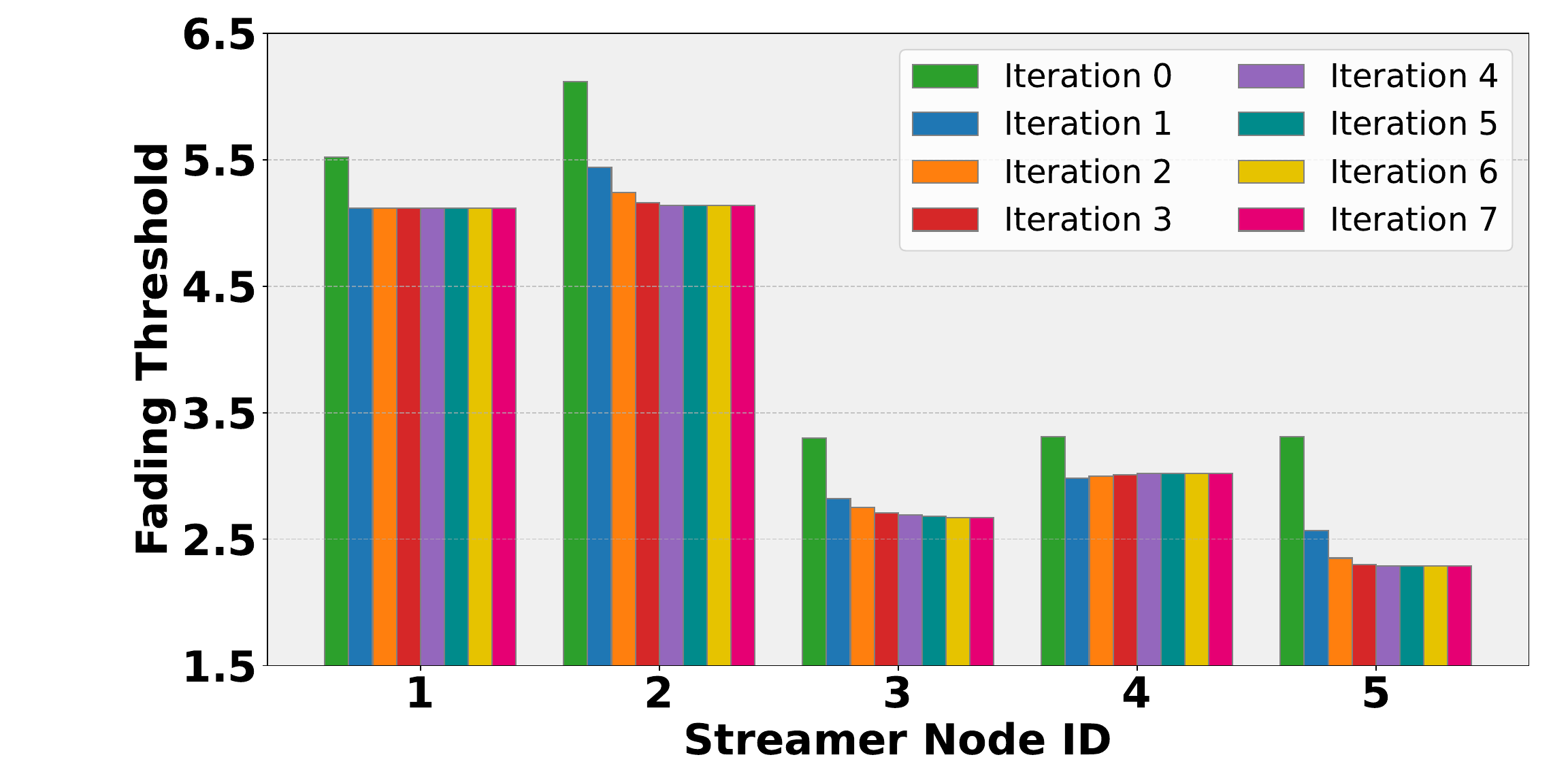}
        \caption{JDVT-EC algorithm convergence for $\boldsymbol{\beta}$.}
        \label{jdvt_ec_beta_perf}
    \end{minipage}%
    \begin{minipage}{.33\linewidth}
        \includegraphics[width=6cm, trim={90 10 10 10}, clip]{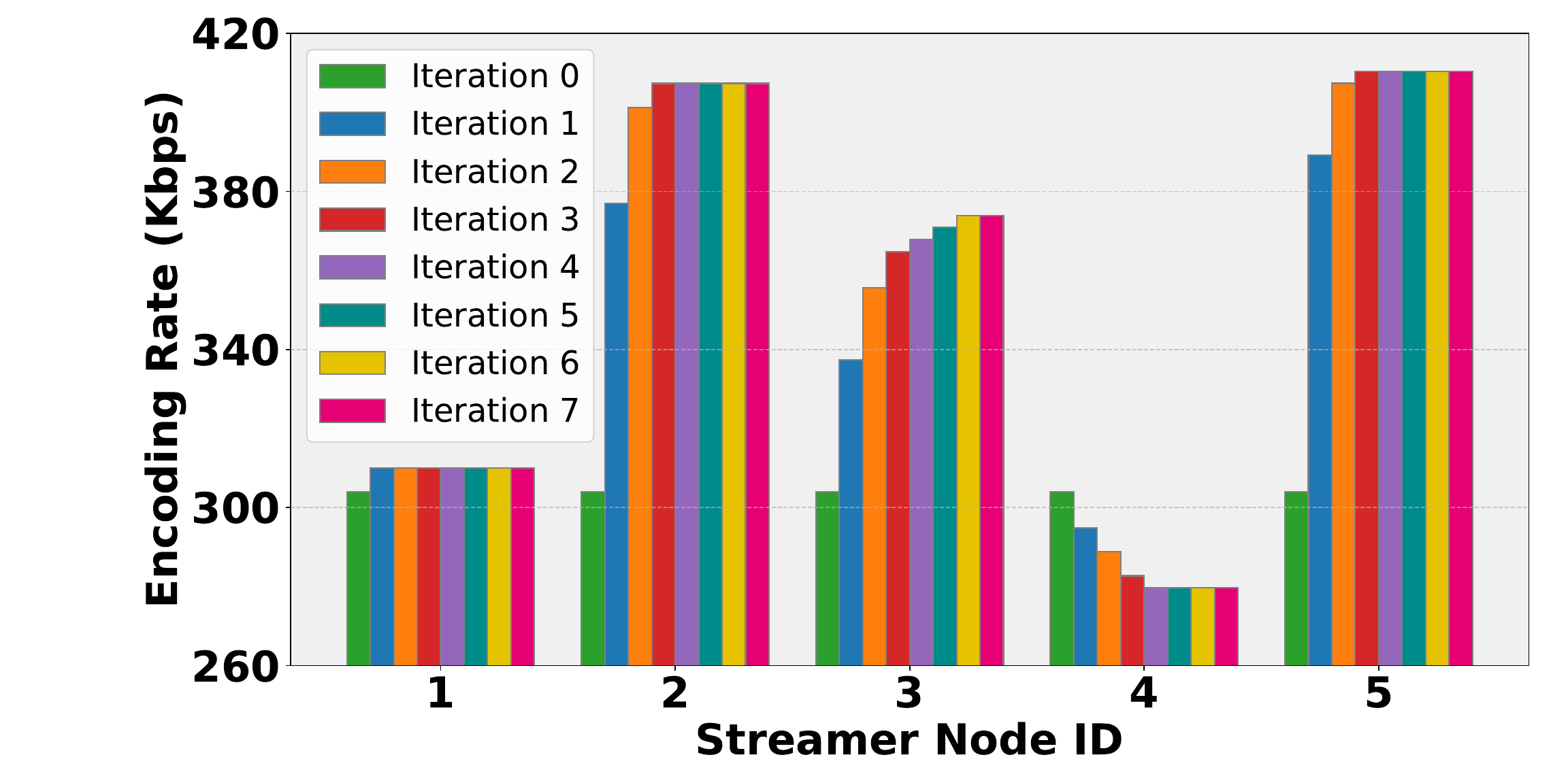}
        \caption{JDVT-EC algorithm convergence for $\boldsymbol{\mathrm{E}}$.}
        \label{jdvt_ec_enc_perf}
    \end{minipage}%
\end{figure*}

\begin{figure*}[t]
    \centering
    \begin{minipage}{.33\linewidth}
        \includegraphics[width=6cm, trim={20 0 100 40}, clip]{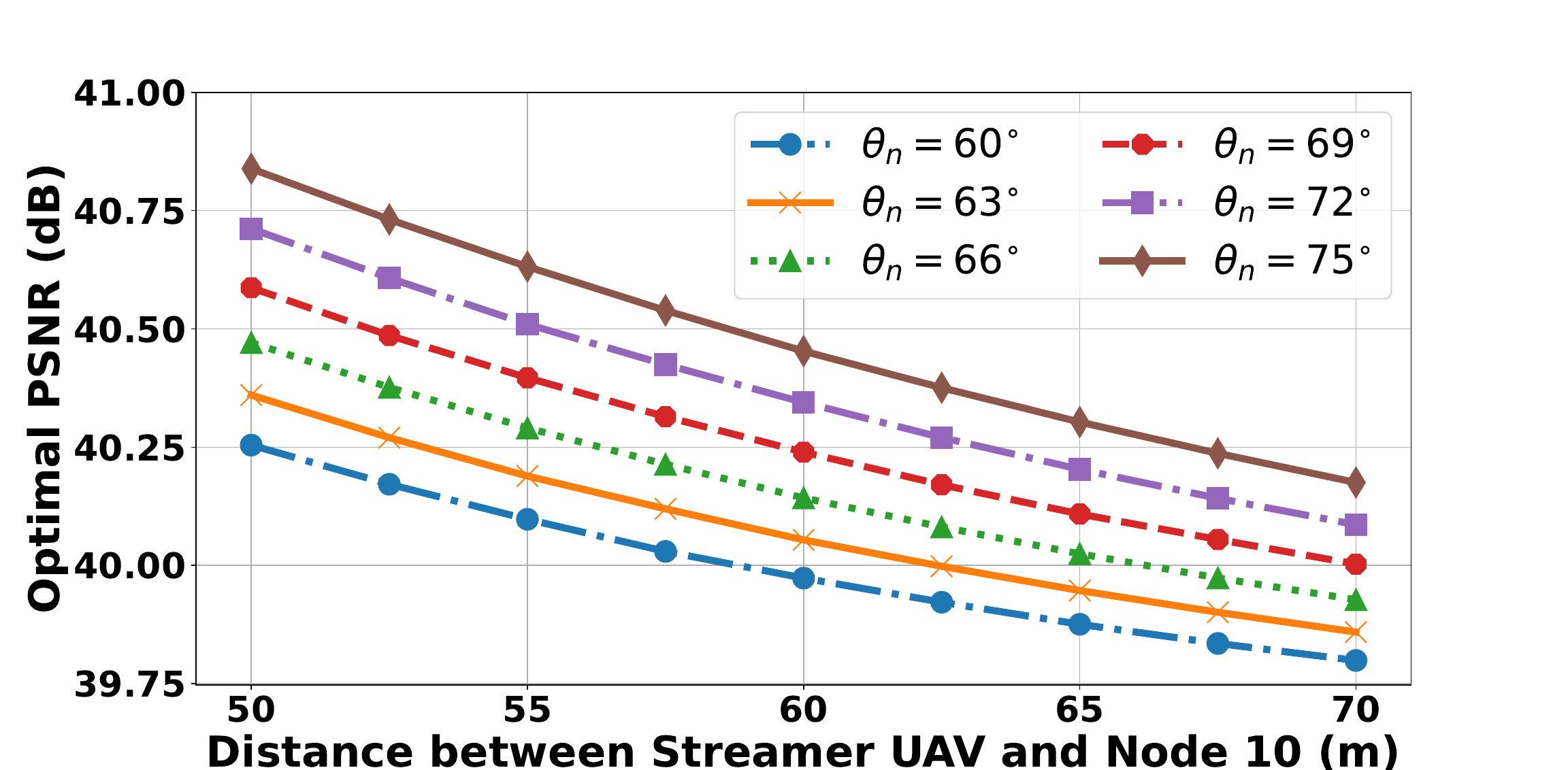}
        \caption{$\mathcal{P}_n^{\star}$ vs. $\theta_n$ for Streamer UAV}
        \label{OptPSNR_Dis}
    \end{minipage}%
    \begin{minipage}{.33\linewidth}
        \includegraphics[width=6cm, trim={20 0 100 40}, clip]{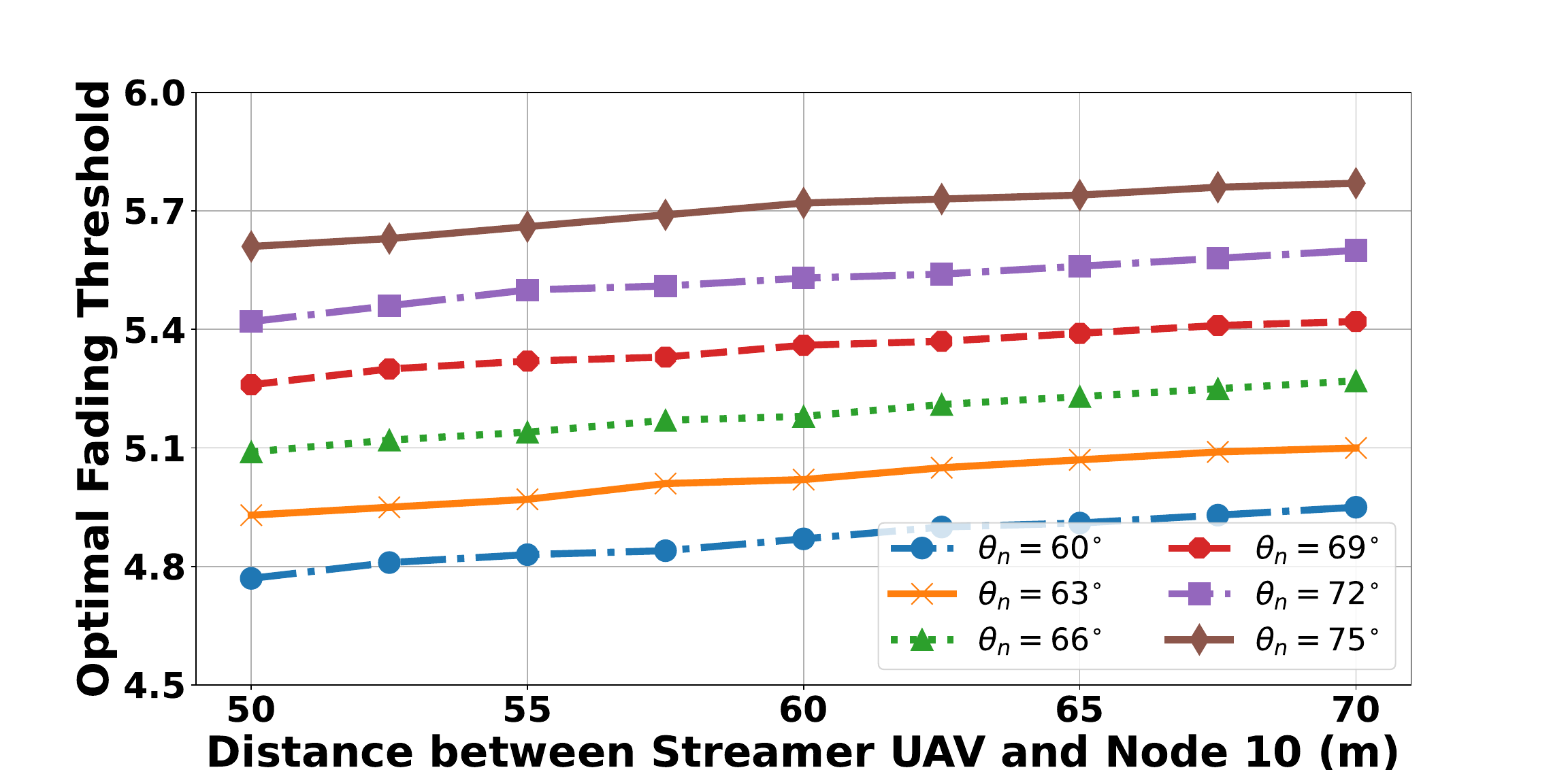}
        \caption{$\beta_n^{\star}$ vs. $\theta_n$ for Streamer UAV}
        \label{OptBeta_Dis}
    \end{minipage}%
    \begin{minipage}{.33\linewidth}
        \includegraphics[width=6cm, trim={20 0 100 40}, clip]{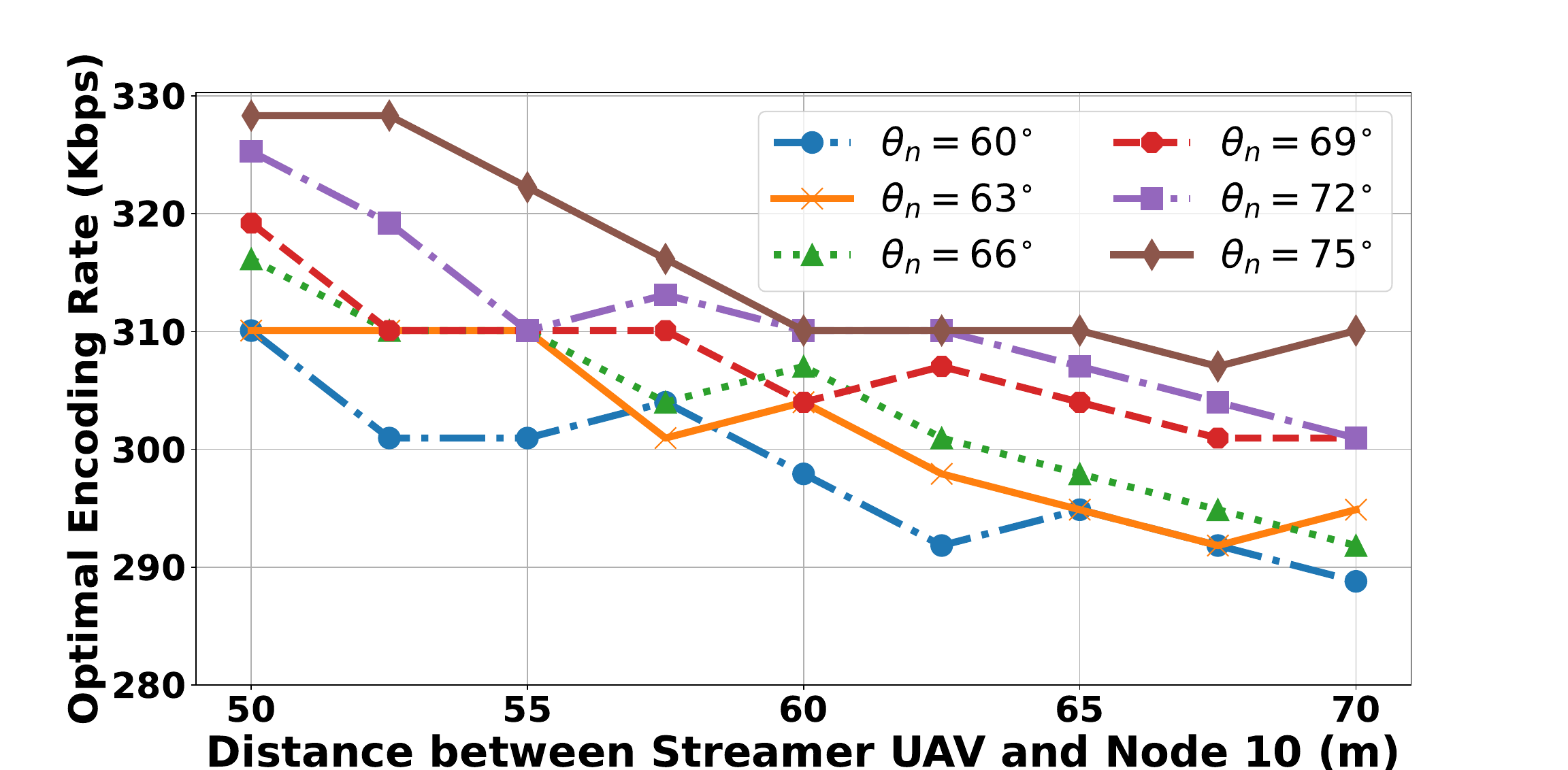}
        \caption{$E_n^{\star}$ vs. $\theta_n$ for Streamer UAV}
        \label{OptLambda_Dis}
    \end{minipage}%
\end{figure*}

Fig. \ref{PSNR_Sen} shows the optimal PSNR $\mathcal{P}_i^{\star}$ from the JDVT-EC algorithm for streamer nodes (nodes 1 to 5) by different sensitivity parameters $s_n$. Clearly, as the sensitivity parameter $s_n$ increases for streamer nodes, the optimal PSNR $\mathcal{P}_i^{\star}$ decreases since the impact of the packet loss distortion $\mathrm{D}_n^{lss}$ increases on the overall distortion $\mathrm{D}_n$. Also, the degradation of the optimal PSNR $\mathcal{P}_i^{\star}$ for streamer nodes with lower optimal PSNR $\mathcal{P}_i^{\star}$ (nodes 1 and 4) is higher than other streamers (nodes 2, 3, and 5). Similarly, the optimal PSNR $\mathcal{P}_i^{\star}$ by optimizing fading threshold $\beta_i$ and encoding rate $E_i$ is represented for different SINR thresholds $\gamma_{th}$ in Fig. \ref{PSNR_SINRth}. Indeed, as the SINR threshold $\gamma_{th}$ increases, the streamer nodes suffer a higher transmission error, and thus the packet loss distortion $\mathrm{D}_n^{lss}$ increases. As a result, the optimal PSNR $\mathcal{P}_i^{\star}$ from the JDVT-EC algorithm degrades by increasing the SINR threshold $\gamma_{th}$. Moreover, the impact of changing SINR threshold $\gamma_{th}$ is negligible on node 4 as an NLoS ground streamer with low optimal PSNR $\mathcal{P}_i^{\star}$, which streams in a crowded area.

\subsection{Algorithm Convergence Performance}
In Fig.~\ref{dtc_perf}, the convergence performance of the DTC algorithm is shown. The results demonstrate that how the fading threshold \( \boldsymbol{\beta} \) is refined until  the optimal value \( \boldsymbol{\beta}^{\star} \) is found (iteration 4). In particular, the algorithm starts from the maximum fading threshold \( \boldsymbol{\beta}^{\max} \) (iteration 0), and then obtains the selfish threshold values \( \boldsymbol{\beta}^{\text{slf}} \) in iteration 1 before finding the optimal threshold \( \boldsymbol{\beta}^{\star} \). Similarly, Figs. \ref{jdvt_ec_beta_perf} and \ref{jdvt_ec_enc_perf}, respectively, illustrate the convergence of fading threshold $\boldsymbol{\beta}$ and encoding rate $\boldsymbol{\mathrm{E}}$ achieved by the JDVT-EC algorithm for streamer nodes. We note that the fading threshold $\boldsymbol{\beta}$ starts at the maximum fading threshold $\boldsymbol{\beta}^{max}$ (iteration 0), and the encoding rate $\boldsymbol{\mathrm{E}}$ starts (iteration 0) at $304$ Kbps, which corresponds to $\lambda_n = 100$. From the results, we note that all streamer nodes reach a consensus (optimal values) on the fading threshold $\boldsymbol{\beta}$ and encoding rate $\boldsymbol{\mathrm{E}}$ within seven iterations.

\subsection{Spatial Performance}
In this part, we investigate the performance of our algorithm as the location of the streamer UAV changes. Due to space constraints, this subsection focuses solely on the streamer UAV's location, while all other nodes remain at fixed positions. In Fig. \ref{PSNR_Heatmap}, the optimal PSNR $\mathcal{P}_n^{\star}$ heatmap for 25 different locations in 100$\times$100 $\text{m}^2$ area is represented for streamer UAV at $50$ m altitude while the optimal fading threshold $\beta_n^{\star}$ and optimal encoding rate $E_n^{\star}$ are shown inside of each squared location. As the streamer UAV gets closer to its associated ground node 10 (-10.22 m, -29.74 m), it can efficiently serve node 10 in the presence of interferer nodes and experiences a higher optimal PSNR $\mathcal{P}_n^{\star}$. Moreover, the streamer UAV increases its optimal encoding rate $E_n^{\star}$ due to good channel condition and optimal fading threshold $\beta_n^{\star}$ since the range of fading threshold extends near node 10. 

In Fig. \ref{OptPSNR_Dis}, the optimal PSNR $\mathcal{P}^{\star}_n$ by the JDVT-EC algorithm is indicated for different distances \textbf{d} = \{50, 52.5, 55, 57.5, 60, 62.5, 65, 67.5, 70\} and elevation angles $\Theta$ = \{60, 63, 66, 69, 72, 75\} between the streamer UAV and node 10. Increasing the distance between the streamer UAV and its ground node (i.e., node 10) gradually degrades the optimal PSNR $\mathcal{P}^{\star}_n$. By increasing the elevation angle $\theta_n$, the streamer UAV establishes a better LoS channel with node 10 since the LoS probability $\mathbb{P}_{L}(d_i)$ increases; thus, the optimal PSNR $\mathcal{P}^{\star}_n$ increases. Overall, the spatial average PSNR $\Bar{\mathcal{P}}_n$ for all distances \textbf{d} and elevation angles $\Theta$ when the streamer UAV serving in this specific spatial area would be $\Bar{\mathcal{P}}_n = 40.22$ dB. Figs. \ref{OptBeta_Dis} and \ref{OptLambda_Dis} illustrate the optimal fading threshold $\beta_n^{\star}$ and optimal encoding rate $E_n^{\star}$ for the streamer UAV by changing the distance $d_n$ and elevation angle $\theta_n$ between the streamer UAV and node 10. 
As the distance increases, the optimal fading threshold $\beta_n^{\star}$ increases. Conversely, the optimal encoding rate $E_n^{\star}$ generally drops due to degraded channel conditions.
 
\section{Conclusion} \label{Conclusion}
In this paper, we investigated the problem of distributed transmission control and video streaming optimization for UAVs operating in unlicensed spectrum bands. We developed an analytical framework that jointly considers cross-layer parameters, including the channel parameters at the PHY layer, queuing parameters at the MAC layer, and video encoding rate at the application layer. Using this framework, we studied the throughput and PSNR according to the overall packet loss $\mathbb{P}_n^{lss}(\boldsymbol{\beta})$ and lossy video compression distortion $\mathrm{D}_n^{cmp}(\mathrm{E}_n)$. In our proposed solution, we introduced two algorithms, namely the DTC algorithm for distributed transmission control and the JDVT-EC algorithm for joint distributed video transmission and encoder control to optimize the video encoding rate $\mathrm{E}_n$ and the fading threshold $\beta_n$ for each source node. The objective function is to maximize the average throughput and PSNR over all nodes. Through extensive numerical evaluations, we demonstrated the efficacy of our algorithms and verified that they consistently achieved optimal solutions. As future work, we will extend our framework to incorporate other parameters, including transmit power and sub-channel allocation. Furthermore, we will investigate the scenarios in which several non-cooperative ground-aerial nodes operate in partially-observable network conditions.

\section*{Acknowledgment}
The material is based upon work supported by NASA under award No(s) 80NSSC20M0261 and NSF grants 1948511, 1955561, 2212565, and 2323189. Any opinions, findings, conclusions, or recommendations expressed in this material are those of the author(s) and do not necessarily reflect the views of NASA and NSF.

\appendices

\section{Proof of Buffer Overflow Model}
\label{bufferoverflowmodel}
\noindent
As mentioned, the probability of exceeding the buffer capacity in a certain state $i$ can be defined as:
\begin{align}
\begin{aligned}
\overline{\mathbb{P}_{i,i+1}} & = \mathbb{P}[x_1+...+x_{i+1} > b_n | x_1+...+x_i \le b_n]
\\ &
= \frac{\int_0^{b_n}\mathbb{P}[x_{i+1} > b_n - x]f_{x_1+...+x_i}(x)dx}{\mathbb{P}[x_1+...+x_i \le b_n]},
\end{aligned}
\end{align}
where $f_{x_1+...+x_i}(x)$ represents the PDF of an i-Erlang distribution. Consequently, the complement of $\overline{\mathbb{P}_{i,i+1}}$ without occurring buffer overflow, can be expressed as:
\begin{align}
\mathbb{P}_{i,i+1} = 1 - \overline{\mathbb{P}_{i,i+1}} = \frac{1 - \sum_{j=0}^i \frac{\Tilde{b}_n^j}{j!}e^{-\Tilde{b}_n}}{1 - \sum_{j=0}^{i-1} \frac{\Tilde{b}_n^j}{j!}e^{-\Tilde{b}_n}}.
\end{align}
Based on the Markov chain, the local balance equation is given by $\pi_{i+1} = \rho_n(\beta_n) \mathbb{P}_{i,i+1}\pi_{i}$,
where $\rho_n(\beta_n) = \frac{\lambda_nT^{slt}_n}{\mu_n(\beta_n)}$ denotes the offered load. Then, $\pi_i$ can be derived as:
\begin{align*}
\pi_i = \rho_n^i(\beta_n) \Big(\prod_{j=0}^{i-1}\mathbb{P}_{j,j+1}\Big)\pi_0 = \rho_n^i(\beta_n) \Big(1-\sum_{j=0}^{i-1}\frac{\Tilde{b}_n^j}{j!}e^{-\Tilde{b}_n}\Big)\pi_0. 
\end{align*}
Now, utilizing $\sum_{j=0}^\infty \pi_j = 1$ and taking into account that $\rho_n(\beta_n) < 1$, $\pi_0$ can be calculated as:
\begin{align}
\pi_0 = \frac{1-\rho_n(\beta_n)}{1-\rho_n(\beta_n) e^{-\Tilde{b}_n(1-\rho_n(\beta_n))}}. 
\end{align}
Finally, the probability of buffer overflow is approximated as:
\begin{align*}
\mathbb{P}_n^{ov}(\beta_n) \approx \sum_{i=0}^\infty \overline{\mathbb{P}_{i,i+1}}\pi_i = \frac{(1-\rho_n(\beta_n))e^{-\Tilde{b}_n(1-\rho_n(\beta_n))}}{1-\rho_n(\beta_n) e^{-\Tilde{b}_n(1-\rho_n(\beta_n))}}. 
\end{align*}

\section{Parameters of Log-normal Distribution}
\label{log-normaldistribution}
\noindent
As mentioned, $\mu(\boldsymbol{\beta}_{-n})$ and $\sigma(\boldsymbol{\beta}_{-n})$ are the location and scale parameters in Log-normal distribution, which are given by:
\begin{align*}
& \mu(\boldsymbol{\beta}_{-n}) = \ln(\Tilde{\mathbb{E}}[I_n^f(\boldsymbol{\beta}_{-n})]) - \frac{1}{2} \ln\Big(1 + \frac{\Tilde{\mathbb{D}}[I_n^f(\boldsymbol{\beta}_{-n})]}{(\Tilde{\mathbb{E}}[I_n^f(\boldsymbol{\beta}_{-n})])^2} \Big),
\\ & 
\sigma(\boldsymbol{\beta}_{-n}) = \sqrt{\ln\Big(1 + \frac{\Tilde{\mathbb{D}}[I_n^f(\boldsymbol{\beta}_{-n})]}{(\Tilde{\mathbb{E}}[I_n^f(\boldsymbol{\beta}_{-n})])^2}\Big)},
\end{align*}
where $\Tilde{\mathbb{E}}[I_n^f(\boldsymbol{\beta}_{-n})]$ and $\Tilde{\mathbb{D}}[I_n^f(\boldsymbol{\beta}_{-n})]$ are the first and second order moments of $I_n^f(\boldsymbol{\beta}_{-n})$, which are shown as:
\begin{align*}
& \Tilde{\mathbb{E}}[I_n^f(\boldsymbol{\beta}_{-n})] = \sum_{m\in \boldsymbol{N}\backslash n} \mathbb{E}\bigl[P_m(\hat{h}_{m}^f)^2(\Tilde{h}_{m}^f)^2\alpha_m^f(\beta_m)\bigl] 
\\ &
= \sum_{m\in \boldsymbol{N}\backslash n} \frac{P_m(\hat{h}_{m}^f)^2 \mu_m(\beta_m)}{|\boldsymbol{F}|} \int_{\beta_m}^\infty x^2 \mathrm{f}_{\Tilde{h}_{m}^f} (x) dx,
\\
& \Tilde{\mathbb{D}}[I_n^f(\boldsymbol{\beta}_{-n})] = \sum_{m\in \boldsymbol{N}\backslash n} \Big(\frac{P_m(\hat{h}_{m}^f)^2 \mu_m(\beta_m)}{|\boldsymbol{F}|}\Big)^2 \times
\\ &
\int_{\beta_m}^\infty x^4 \mathrm{f}_{\Tilde{h}_{m}^f} (x) dx + \sum_{m_1 \neq m_2} \frac{P_{m_1}(\hat{h}_{m_1}^f)^2 \mu_{m_1}(\beta_{m_1})}{|\boldsymbol{F}|} \times
\\ &
\int_{\beta_{m_1}}^\infty  x^2 \mathrm{f}_{\Tilde{h}_{m_1}^f} (x) dx \frac{P_{m_2}(\hat{h}_{m_2}^f)^2 \mu_{m_2}(\beta_{m_2})}{|\boldsymbol{F}|} \times 
\\ &
\int_{\beta_{m_2}}^\infty  x^2 \mathrm{f}_{\Tilde{h}_{m_2}^f} (x) dx - (\Tilde{\mathbb{E}}[I_n^f(\boldsymbol{\beta}_{-n})])^2.
\end{align*}
Here, $\mu_m(\beta_m)$ and $\mathrm{f}_{\Tilde{h}_{m}^f} (x)$ depend on the type of channel between the interferer node $m$ and the destination.

\section{Concavity of $\mathbb R_n$ and $\mathcal{P}_n$ with Respect to $\beta_n$} \label{concave_beta}
\noindent
In this section, we provide proofs of concavity of $\mathbb R_n$ and $\mathcal{P}_n$ with respect to $\beta_n$. It should be noted that while our distributed optimization problem is solved numerically, establishing such concavity results guarantees the convergence to the global optimum. To this end, we note that based on Eq.~\eqref{Rn} and Fig.~\ref{exact_approx}, the probability of overall loss $\mathbb{P}_n^{lss}(\boldsymbol{\beta})$ can be approximated by:
\begin{align*}
\small
\begin{aligned}
\mathbb{P}_n^{lss}(\boldsymbol{\beta}) \approx \mathbb{P}_n^{ov}(\beta_n) + \mathbb{P}_n^{dly}(\beta_n) + \mathbb{P}_n^{err}(\boldsymbol{\beta}).
\end{aligned}
\end{align*}
\begin{figure}[t]
\includegraphics[width=\linewidth, trim={40 0 100 50}, clip]{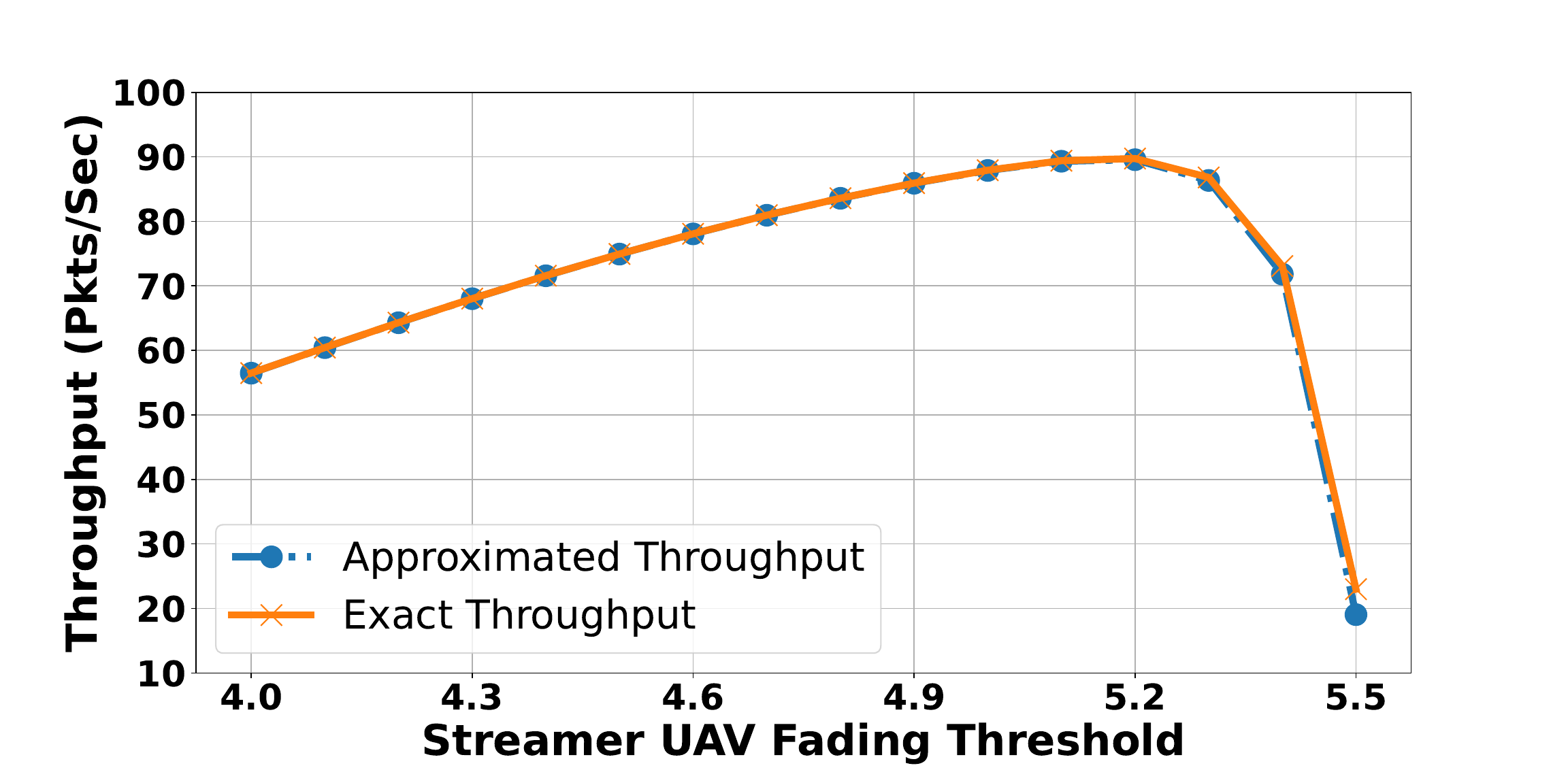}
\centering
\caption{Approximated and exact throughput $\mathbb R_n(\boldsymbol{\beta})$ values by different streamer UAV fading thresholds.}
\label{exact_approx}
\end{figure}
Accordingly, we only need to prove the convexity of $\mathbb{P}_n^{lss}$ with respect to $\beta_n$ i.e. $\frac{\partial^2 \mathbb{P}_n^{lss}}{\partial \beta_n^2} > 0$. Thus, we have:
\begin{equation}
\frac{\partial \mathbb{P}^{lss}_n}{\partial \beta_n} = \frac{\partial \mathbb{P}^{err}_n}{\partial \beta_n} + \frac{\partial \mathbb{P}^{dly}_n}{\partial \beta_n}.
\end{equation}
For the sake of simplicity, $\mathbb{P}_n^{ov}(\beta_n)$ is ignored due to its negligible value as indicated in Fig. \ref{Plss(node10)-Beta}. Hence, $\frac{\partial \mathbb{P}_n ^{lss}}{\partial \beta_n}$ can be expressed as:
\begin{align*}
\frac{\partial \mathbb{P}_n ^{lss}}{\partial \beta_n} & = -\mathrm{f}_{\Tilde{h}_n^f} (\beta_n) v_n(X,\boldsymbol{\beta}_{-n}) 
\\ &
+ \frac{|\boldsymbol{F}|T_n^{th}}{T^{slt}_n} 
\mathbb{P}_n^{dly}(\beta_n) 
\mathcal{F}_{\Tilde{h}_{n}^f} (\beta_n)^{|\boldsymbol{F}|-1} \mathrm{f}_{\Tilde{h}_n^f} (\beta_n),
\end{align*}
where $X = \frac{P_n \hat{h}_n^2}{\gamma_{th}} \beta_n^2 - \sigma^2$. Then,
\begin{equation}
\frac{\partial^2 \mathbb{P}^{lss}_n}{\partial \beta_n^2} = \frac{\partial^2 \mathbb{P}^{err}_n}{\partial \beta_n^2} + \frac{\partial^2 \mathbb{P}^{dly}_n}{\partial \beta_n^2}.
\end{equation}
Here, $\frac{\partial^2 \mathbb{P}_n^{err}}{\partial \beta_n^2}$ and $\frac{\partial^2 \mathbb{P}_n ^{dly}}{\partial \beta_n^2}$ can be calculated as:
\begin{align*}
\frac{\partial^2 \mathbb{P}_n^{err}}{\partial \beta_n^2} & = e^{-\frac{\beta_n^2+b^2}{2}} \Big((\beta_n^2 - 1)I_0(\beta_n b) - \beta_n bI_1(\beta_nb) \Big)
\\ &
\times v_n(X,\boldsymbol{\beta}_{-n}) - \mathrm{f}_{\Tilde{h}_{n}^f} (\beta_n) \frac{\partial v_n(X,\boldsymbol{\beta}_{-n})}{\partial\beta_n},
\end{align*}
where $\frac{\partial v_n(X,\boldsymbol{\beta}_{-n})}{\partial\beta_n} < 0$ and $0 \le v_n(X,\boldsymbol{\beta}_{-n}) \le 1$. Using the upper bound of fading threshold and $I_0(\beta_nb) > I_1(\beta_nb)$, $(\beta_n^2 - 1)I_0(\beta_n b) > \beta_n bI_1(\beta_nb)$, we have $\frac{\partial^2 \mathbb{P}_n^{err}}{\partial \beta_n^2} > 0$. Then,
\begin{align*}
& \frac{\partial^2 \mathbb{P}_n ^{dly}}{\partial \beta_n^2} = \frac{\partial \mathbb{P}_n ^{dly}}{\partial \beta_n} \frac{|\boldsymbol{F}|T_n^{th}}{T^{slt}_n} 
\mathcal{F}_{\Tilde{h}_{n}^f} (\beta_n)^{|\boldsymbol{F}|-1} \mathrm{f}_{\Tilde{h}_{n}^f} (\beta_n) + 
\\ &
\mathbb{P}_n^{dly}(\beta_n) \frac{|\boldsymbol{F}|T_n^{th}}{T^{slt}_n}
(|\boldsymbol{F}|-1) \mathcal{F}_{\Tilde{h}_{n}^f} (\beta_n)^{|\boldsymbol{F}|-2} \mathrm{f}_{\Tilde{h}_{n}^f} (\beta_n)^2
\\ &
+ \frac{\partial \mathrm{f}_{\Tilde{h}_n^f} (\beta_n)}{\partial\beta_n} \frac{|\boldsymbol{F}|T_n^{th}}{T^{slt}_n} 
\mathbb{P}_n^{dly}(\beta_n) 
\mathcal{F}_{\Tilde{h}_{n}^f} (\beta_n)^{|\boldsymbol{F}|-1},
\end{align*}
where $\frac{\partial \mathbb{P}_n ^{dly}}{\partial \beta_n} > 0$. Assuming $\frac{|\boldsymbol{F}|T_n^{th}}{T^{slt}_n} \gg 1$ to prove $\frac{\partial^2 \mathbb{P}_n ^{dly}}{\partial \beta_n^2} > 0$, need to demonstrate:
\begin{align*}
\frac{\partial \mathbb{P}_n^{dly}}{\partial \beta_n} \Big(1 + \frac{1-\beta_n^2}{\beta_n} + \frac{bI_1(\beta_nb)}{I_0(\beta_nb)} \Big) + |\boldsymbol{F}| - 1 > 0.
\end{align*}
This inequality can be numerically verified as long as the upper bound of fading threshold holds. 

\section{Concavity of $\mathcal{P}_n$ with Respect to $\mathrm{E}_n$}
\label{concave_en}
\noindent
To prove the concavity of $\mathcal{P}_n$ with respect to $E_n$, it is sufficient to prove the convexity of $\ln \mathrm{D}_n$ with respect to $E_n$ as follows:
\begin{equation} \label{ineq}
\frac{\partial^2 \mathrm{D}_n}{\partial \mathrm{E}_n^2} \mathrm{D}_n - 2\Big(\frac{\partial \mathrm{D}_n}{\partial \mathrm{E}_n}\Big)^2 \ge 0,
\end{equation}
where the first and second derivatives of $\mathrm{D}_n$ are provided as: 
\begin{align*}
& \frac{\partial \mathrm{D}_n}{\partial \mathrm{E}_n} = -\frac{\theta_0}{(\mathrm{E}_n - E_{0})^2} + \frac{s_n T_n^{th}}{L_n} \mathbb{P}_{n}^{dly}(\beta_n),
\\ & 
\frac{\partial^2 \mathrm{D}_n}{\partial \mathrm{E}_n^2}  = \frac{2\theta_0}{(\mathrm{E}_n - E_{0})^3} + s_n (\frac{ T_n^{th}}{L_n})^2 \mathbb{P}_{n}^{dly}(\beta_n).
\end{align*}
Here, it can be obtained:
\begin{align*}
& \frac{\partial^2 \mathrm{D}_n}{\partial \mathrm{E}_n^2} \mathrm{D}_n = \frac{2\theta_0 D_0}{(\mathrm{E}_n - E_{0})^3} + s_n D_0 (\frac{ T_n^{th}}{L_n})^2 \mathbb{P}_{n}^{dly}(\beta_n) +
\\ &
\frac{2\theta_0^2}{(\mathrm{E}_n - E_{0})^4} + \frac{s_n \theta_0}{\mathrm{E}_n - E_{0}} (\frac{T_n^{th}}{L_n})^2 \mathbb{P}_{n}^{dly}(\beta_n) + 
\\ &
\frac{2 s_n \theta_0 \mathbb{P}_{n}^{lss}(\boldsymbol{\beta})}{(\mathrm{E}_n - E_{0})^3} + (\frac{s_n T_n^{th}}{L_n})^2 \mathbb{P}_{n}^{dly}(\beta_n) \mathbb{P}_{n}^{lss}(\boldsymbol{\beta}),
\\ &
2\Big(\frac{\partial \mathrm{D}_n}{\partial \mathrm{E}_n}\Big)^2 = \frac{2\theta_0^2}{(\mathrm{E}_n - E_{0})^4} + 2 (\frac{s_n T_n^{th}}{L_n})^2 \mathbb{P}_{n}^{dly}(\beta_n)^2 
\\ &
- \frac{4 s_n \theta_0 T_n^{th} \mathbb{P}_{n}^{dly}(\beta_n)}{(\mathrm{E}_n - E_{0})^2 L_n},
\end{align*}
since $\mathbb{P}_{n}^{lss}(\boldsymbol{\beta}) > \mathbb{P}_{n}^{dly}(\beta_n)$, we can get the following inequality:
\begin{align*}
& \frac{2\theta_0 D_0}{(\mathrm{E}_n - E_{0})^3} + s_n D_0 (\frac{ T_n^{th}}{L_n})^2 \mathbb{P}_{n}^{dly}(\beta_n) +
\\ &
\frac{s_n \theta_0}{\mathrm{E}_n - E_{0}} (\frac{T_n^{th}}{L_n})^2 \mathbb{P}_{n}^{dly}(\beta_n) + \frac{2 s_n \theta_0 \mathbb{P}_{n}^{lss}(\boldsymbol{\beta})}{(\mathrm{E}_n - E_{0})^3}
\\ &
+ \frac{4 s_n \theta_0 T_n^{th} \mathbb{P}_{n}^{dly}(\beta_n)}{(\mathrm{E}_n - E_{0})^2 L_n} - (\frac{s_n T_n^{th}}{L_n})^2 \mathbb{P}_{n}^{dly}(\beta_n)^2 \ge 0,
\end{align*}
where $\theta_0 \gg s_n$, $\mathbb{P}_{n}^{dly}(\beta_n) > \mathbb{P}_{n}^{dly}(\beta_n)^2$, and $\frac{T_n^{th}}{L_n} \ll 1$. Thus, this inequality can be proved.

\begin{algorithm}[t]
\caption{Local Coordinate Search (LCS)}
\label{alg:LCS}
\begin{algorithmic}[1]
\Function{LCS}{$x$, $\boldsymbol{stp}$, $\mathtt{f}^{\star}$, $x^{max}$}
\State $bool \gets true$, $stp_x \gets \tau_x \times \omega_x$, $flg \gets 0$
\State \textbf{global} $x, stp_x, flg, \mathtt{f}^{\star}, bool$
\While{$bool$}
\State \Call{FBS}{$flg$, 1, $\boldsymbol{stp}$, $x^{max}$} $\#$ForwardSearch
\State \Call{FBS}{$flg$, -1, $\boldsymbol{stp}$, $x^{max}$} $\#$BackwardSearch
\EndWhile
\State \textbf{return} $x$
\EndFunction
\Function{FBS}{$k$, $sgn$, $\boldsymbol{stp}$, $x^{max}$}
    \While{$flg = k$ \textbf{and} $bool$} $stp_x \gets stp_x \times \omega^{-1}_{x}$
        \State $x^{old} \gets x$, $x \gets \Call{clip}{x + stp_x \times sgn, 0, x^{max}}$
        \If{$\mathtt{f}(x) > \mathtt{f}^{\star}$} $\mathtt{f}^{\star} \gets \mathtt{f}(x)$
        \Else: $flg \gets k + 1$, $x \gets x^{old}$
        \EndIf
    \EndWhile
    \While{$flg = k + 1$ \textbf{and} $bool$} $stp_x \gets stp_x \times \omega_x$
        \If{$stp_x < \tau_x$} $bool \gets false$
            \State \textbf{break}
        \EndIf
        \State $x^{old} \gets x$, $x \gets \Call{clip}{x + stp_x \times sgn, 0, x^{max}}$
        \If{$\mathtt{f}(x) > \mathtt{f}^{\star}$} $\mathtt{f}^{\star} \gets \mathtt{f}(x)$
        \Else: $x \gets x^{old}$
        \If{$stp_x = \tau_x$} 
            \State $flg \gets k + 2 \times sgn$, $stp_x \gets stp_x \times \omega_x$
        \EndIf
        \EndIf
    \EndWhile
\EndFunction
\Function{clip}{$x$, $x^{min}$, $x^{max}$}
\If{$x^{min} < x < x^{max}$} \textbf{return} $x$
\ElsIf{$x \ge x^{max}$} \textbf{return} $x^{max}$
\ElsIf{$x \le x^{min}$} \textbf{return} $x^{min}$
\EndIf
\EndFunction
\end{algorithmic}
\end{algorithm}

\section{Local Coordinate Search (LCS) Algorithm}
\label{LCS}
\noindent
The Local Coordinate Search (LCS) algorithm aims to determine the optimal value of $x = \{\beta_n, \lambda_n\}$ for each node in each iteration. In this algorithm, the LCS function uses the Forward-Backward Search (FBS) function, which handles the coordinate search direction ($sgn$), step size ($stp_x$), and where to terminate the algorithm ($bool$). The CLIP function limits the decision variable $x$ to a range determined by $x^{min}$ and $x^{max}$. Furthermore, $\boldsymbol{stp} = \{\boldsymbol{stp}_{\beta}, \boldsymbol{stp}_{\lambda} \} = \{\{\omega_{\beta}, \tau_{\beta}\}, \{\omega_{\lambda}, \tau_{\lambda}\} \}$ includes the step divider $\omega_x$, which changes the step size proportionally, and step accuracy $\tau_x$, which controls the accuracy of the decision variable and stops the algorithm according to the step size.

In Table \ref{time_complexity}, the time complexity of the proposed algorithms is presented, with the LCS algorithm serving as the core component. Therefore, we first analyze the complexity of the LCS algorithm, after which the complexity of the other algorithms can be readily determined. To this end, we note that the while 
loop $flg = k$ (line 11 in the LCS algorithm) terminates at most when $x$ meets $x^{max}$ or 0. Furthermore, the while loop $flg = k + 1$ (in line 17) terminates when by proportionally changing $stp_x$, the condition $stp_x < \tau_x$ is met for $\omega_x \in (0, 1)$. By combining these two, the worst-case complexity of the LCS algorithm can be derived as $\mathcal{O}(\log_{\omega^{-1}_x}(\frac{x^{max}}{\tau_x})) + \mathcal{O}(\log_{\omega^{-1}_x}(\frac{stp^{max}_x}{\tau_x})) = \mathcal{O}(\log_{\omega^{-1}_x}(\frac{stp^{max}_x \times x^{max}}{\tau_x^2}))$ where $stp^{max}_x = \frac{\tau_x}{\omega_x^{t}}$ (maximum initial step in $flg = k + 1$) in which $t$ is the smallest value satisfying $t \ge \frac{\ln (1 - \frac{x^{max}(1 - \frac{1}{\omega_x})}{\tau_x})}{|\ln \omega_x|}$. Given the time complexity of the LCS algorithm, we derive the worst-case performance of the other algorithms, as reported in Table~\ref{time_complexity} in which $i_{max}$ and $|\boldsymbol{M}|$ are the maximum number of iterations and number of nodes, respectively.

\begin{table}
    \centering
    \caption{Time Complexity of Algorithms}
    \resizebox{\columnwidth}{!}{
    \begin{tabular}{lc}
        \toprule
        Algorithm & Time Complexity \\
        \midrule
            LCS (Alg. \ref{alg:LCS}) & $\mathcal{O}(\log_{\omega^{-1}_x}(\frac{stp^{max}_x \times x^{max}}{\tau_x^2}))$\\
        DTC (Alg. \ref{alg:DTC}) & $\mathcal{O}(i_{max}.|\boldsymbol{M}|.\log_{\omega^{-1}_x}(\frac{stp^{max}_x \times x^{max}}{\tau_x^2}))$ \\
        DVTC (Alg. \ref{alg:DVTC}) & $\mathcal{O}(i_{max}.|\boldsymbol{M}|.\log_{\omega^{-1}_x}(\frac{stp^{max}_x \times x^{max}}{\tau_x^2}))$\\
        DVEC (Alg. \ref{alg:DVEC}) & $\mathcal{O}(|\boldsymbol{M}|.\log_{\omega^{-1}_x}(\frac{stp^{max}_x \times x^{max}}{\tau_x^2}))$ \\
        JDVT-EC (Alg. \ref{alg:JDVT-EC}) & $\mathcal{O}(i_{max}^2.|\boldsymbol{M}|.\log_{\omega^{-1}_x}(\frac{stp^{max}_x \times x^{max}}{\tau_x^2}))$\\
        \bottomrule
    \end{tabular}
    }
    \label{time_complexity}
\end{table}

{\small
\bibliographystyle{IEEEtranN}
\bibliography{ref}
}

\end{document}